\documentclass[]{article}
\usepackage[a4paper, total={7in, 9in}]{geometry}
\usepackage{authblk}

\usepackage{amsfonts}
\usepackage{amsmath}
\usepackage{amssymb}

\usepackage{graphicx}

\usepackage{bm}
\usepackage{cancel}

\usepackage{hyperref}

\usepackage{amsthm}
\newtheorem{theorem}{Theorem}

\newtheorem{proposition}[theorem]{Proposition}

\theoremstyle{definition}
\newtheorem{definition}[theorem]{Definition}
\theoremstyle{remark}

\numberwithin{theorem}{section}

\usepackage{tikz}

\tikzset{
        coord/.style={draw, circle, inner sep=0pt, minimum size=14pt},
        coordlabel/.style={anchor=center, minimum size=14pt, label distance=14pt},
        bcoord/.style={coord, label={[coordlabel]270:#1}},
        lcoord/.style={coord, label={[coordlabel]180:#1}},
        minicoord/.style={draw,circle,inner sep=0pt, minimum size=9pt,execute at begin node=\footnotesize},
        rmatrix/.style={
            matrix of math nodes, 
            nodes={coord,anchor=center}, 
            column sep={20pt,between origins}, 
            row sep={20pt,between origins}
            },
        spin/.style={draw=none, fill, circle, inner sep=2pt, minimum size=0pt},
        wcoord/.style={coord, append after command={(\tikzlastnode)-- +(0.5,0) (\tikzlastnode)-- +(-0.5,0)}},
        zcoord/.style={coord, append after command={(\tikzlastnode)-- +(0,0.5) (\tikzlastnode)-- +(0,-0.5)}}
        }

\newcommand*{\C}{\mathbb{C}}

\newcommand{\col}[1]{\tikz{\node[minicoord]{$#1$};}}

\newcommand*{\rmi}{\mathrm{I}}
\newcommand*{\rmii}{\mathrm{I\!I}}

\title{Higher rank elliptic partition functions
and \\ multisymmetric elliptic functions}

\author[1]{Allan John Gerrard\footnote{a.j.gerrard$\otimes$rs.tus.ac.jp}}
\affil[1]{Department of Physics, Tokyo University of Science}

\author[2]{Kohei Motegi\footnote{kmoteg0$\otimes$kaiyodai.ac.jp}}
\affil[2]{Faculty of Marine Technology, Tokyo University of Marine Science and Technology}

\author[1]{Kazumitsu Sakai\footnote{k.sakai$\otimes$rs.tus.ac.jp}}

\begin{document}

\maketitle

\begin{abstract}
We introduce and investigate a class of $\mathfrak{gl}_{M+1}$ partition functions which
is an extension of the one introduced by Foda-Manabe.
We characterize the partition functions by a nested version of Izergin-Korepin analysis,
and determine the explicit forms, for each of the rational, trigonometric and elliptic versions.
The resulting multisymmetric functions can be regarded
as extensions of the rational, trigonometric and elliptic weight functions.
\end{abstract}

    
    
    
    

\section{Introduction}
Partition functions of integrable lattice models
\cite{Bethe,Baxter,KBI} are not just important objects in statistical physics
but also have deep connections with mathematics.
The first type investigated extensively was the domain wall boundary partition functions
\cite{Ko,Iz,Ku1,Ku2,Tsuchiya}.
More recently, partition functions have been investigated as special functions and symmetric functions.
For the quantum group case, see
\cite{BBF,TarVarSigma,BWZ,BP,vDE,Borodin,Tak,metaplectic,FM,BW,BFHTW} to list a few.
For the dynamical/elliptic case, see
\cite{FV2,TVAst,PRS,Ros,FK,Chinesegroup,Chinesegroup2,Galleasone,GL,Lamers,Motegi,
Aggarwal,Borodin2,Motegi2}.

A class of special functions which are known as weight functions
appeared in the context of quantum integrable systems,
as parts of integrands of solutions to the $q$KZ equations and
as off-shell Bethe wavefunctions in algebraic Bethe ansatz.
For rational/trigonometric weight functions, see
\cite{TVAst,Matsuo,TV,Mimachi,TarVarSigma,RTV1,Sh} for example.
Partition functions corresponding to these special functions first appeared in \cite{Reshetikhin}
and were further studied in \cite{TarVarSigma,BW,FM,Motegi2}.
See also \cite{KT1,KT2} for new expressions of weight functions
and \cite{Smirnov} for an extension to a different direction. 
Higher rank elliptic weight functions were introduced in
\cite{Konno1,Konno2,FRV,RTV2}.
These weight functions also appear as a geometric basis named as stable basis
\cite{MO,AO}, motivated by mathematically formulating \cite{NS1,NS2}.

In this paper, we introduce and investigate a generalization of the type of partition functions
introduced by Foda-Manabe for all rational/trigonometric/elliptic models.
We call the analysis of the partition functions used in this paper nested Izergin-Korepin analysis,
since it is a higher rank version of Izergin-Korepin analysis which was first
used to determine the explicit forms of the domain wall boundary partition functions
of the six-vertex model.
We characterize the partition functions by the nested Izergin-Korepin analysis,
and determine the explicit forms as multisymmetric functions.
The multisymmetric functions introduced generalize the ones by \cite{FM} and \cite{Motegi2} for the elliptic $\mathfrak{gl}_3$ case,
and can be regarded as extensions of the rational, trigonometric and elliptic weight functions.
The analysis used in this paper can treat all three types of models in a unified way.
The nested version of the Izergin-Korepin analysis is similar to \cite{Wheeler}
which Slavnov's determinant formula for the scalar product \cite{Slavnov} was
characterized by using essentially the domain wall boundary partition function as the initial condition
and the scalar product as the final output of the recursion.
Here we use essentially the $\mathfrak{gl}_M$ partition function as the initial condition
and construct a recursion relation between the $\mathfrak{gl}_{M+1}$ partition functions.

This paper is organized as follows.
In the next section, we introduce the partition functions of rational/trigonometric type and give a detailed description
of the labels and symbols. We characterize the partition functions by
constructing a nested version of Korepin's Lemma.
Then we introduce multisymmetric functions and show they satisfy all the properties of Korepin's Lemma.
We give the details of the computations for the rational version.
In section 3, we introduce an elliptic version of the partition functions
and perform the same analysis.
We show in section 4, the correspondence between
the special cases of the elliptic multisymmetric functions introduced in section 3
to the elliptic weight functions by Konno and Rimanyi-Tarasov-Varchenko is explained.

\section{Rational and trigonometric $\mathfrak{gl}_{M+1}$ partition functions}
In this section,
we introduce and study partition functions constructed from $R$-matrices
of the Yangian $Y(\mathfrak{gl}_{M+1})$ and the quantum affine algebra $U_q(\widehat{\mathfrak{sl}}_{M+1})$
\cite{Drinfeld,Jimbo,RTF}.

\subsection{$R$-matrix}

Let $V$ be an $(M+1)$-dimensional vector space spanned by standard basis $e_k$,  $1 \leq k \leq M+1$.
We call the subscripts of the standard basis $1 \leq k \leq M+1$ colors in this paper.

The rational $R$-matrix is
\begin{align}
R(x,y)=\sum_{i=1}^{M+1} (x-y+1) E_{ii} \otimes E_{ii}+\sum_{i \neq j} (x-y) E_{ii} \otimes E_{jj}
+\sum_{i \neq j}E_{ij} \otimes E_{ji},
\end{align}
where $E_{ij}$ are matrix units, satisfying $E_{ij}e_k=\delta_{jk}e_i$.

The $R$-matrix is
graphically represented as Figure \ref{rationalrmatrix-components}.
Here, only the nonzero matrix elements of the rational $R$-matrix
are shown.
 \begin{figure}[ht]
        \centering
        \begin{tabular}{ccc}
            \includegraphics{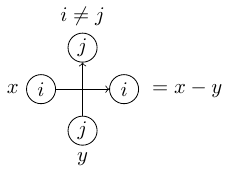}
            &
            \includegraphics{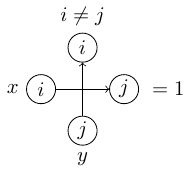}
            &
            \includegraphics{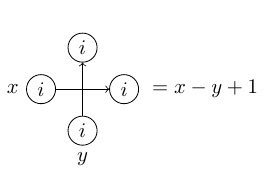}
        \end{tabular}        
        \caption{Matrix elements of the rational $R$-matrix $R(x,y)$.}
\label{rationalrmatrix-components}
    \end{figure}

The $R$-matrix satisfies the Yang-Baxter equation
\begin{align}
R_{23}(x_2-x_3)R_{13}(x_1-x_3)R_{12}(x_1-x_2)=
R_{12}(x_1-x_2)R_{13}(x_1-x_3)R_{23}(x_2-x_3),
\label{rationalYBE}
\end{align}
acting on $V_1 \otimes V_2 \otimes V_3$.
Each $R_{ij}(x,y)$ acts nontrivially on $V_i$ and $V_j$, and on the remaining space as identity, for example $R_{12}(x,y) = R(x,y) \otimes I$.

See Figure \ref{figureyangbaxter} for a graphical description of \eqref{rationalYBE}.
    \begin{figure}[ht]
        \centering
        \includegraphics{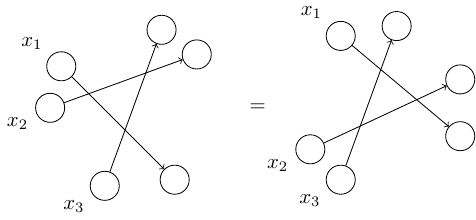}
        \caption{The Yang-Baxter equation \eqref{rationalYBE}.}
\label{figureyangbaxter}
    \end{figure}

\subsection{Partition functions}
We introduce a class of partition functions graphically represented by Figure
\ref{fig:rationalpartitionfunctions}.
This figure is basically the same for the rational, trigonometric and elliptic versions, although we need additional explanation for the elliptic case because of the presence of the
dynamical parameters which will be explained in the next section.

\begin{figure}[ht] 
    \includegraphics[width=\textwidth]{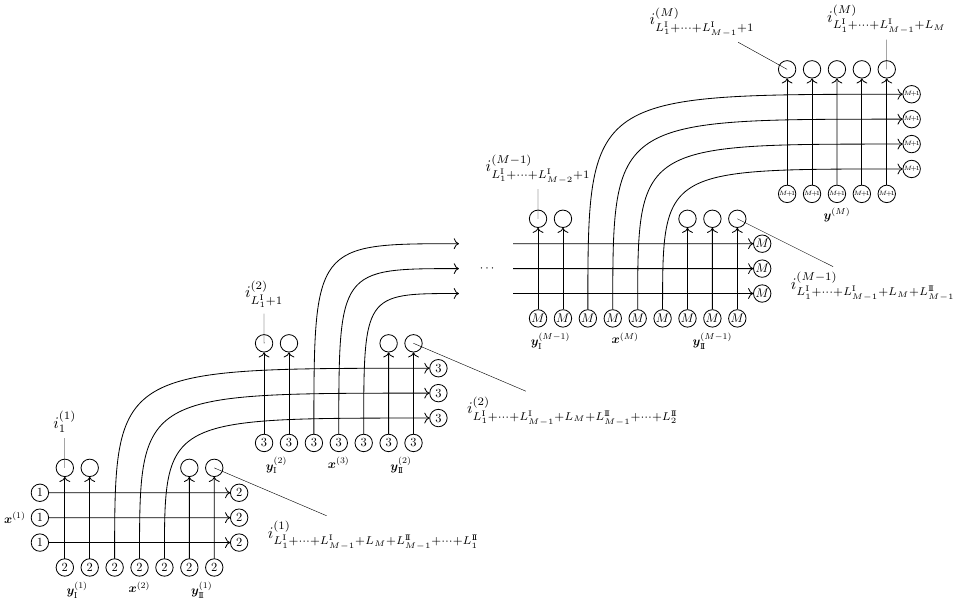}
\caption{Rational and trigonometric $\mathfrak{gl}_{M+1}$ partition functions
\eqref{rationalpartitionfunctions} and \eqref{trigonometricpartitionfunctions}.
For the trigonometric case, the $x$- and $y$-variables are replaced by $u$- and $v$-variables.
}
\label{fig:rationalpartitionfunctions}
\end{figure}

In Figure \ref{fig:rationalpartitionfunctions},
there are $M$ layers.
We call the layer in the bottom-left the \emph{first layer},
the layer northeast to it the \emph{second layer}, and so on.
Each layer consists of horizontal lines and vertical lines, which represent vector spaces.
In keeping with tradition, we call the horizontal lines \emph{auxiliary spaces}
and vertical lines \emph{quantum spaces}.
Note that some of the quantum spaces in the $j$-th layer become auxiliary spaces in the $(j+1)$-th layer.
To each vector space, we also assign a spectral variable.
The set of spectral variables of the auxiliary spaces in the $j$-th layer is denoted by ${\bm x}^{(j)}
=\{x^{(j)}_1,\dots,x^{(j)}_{k_j} \}$ ($k_j=|{\bm x}^{(j)}|$).
Note that the ordering of the variables in the auxiliary spaces can be arbitrary due to the Yang-Baxter relation. 
The spectral variables in the quantum spaces in the $j$-th layer are ${\bm y}_{\rmi}^{(j)}
=\{y_{\rmi,1}^{(j)},\dots,y_{\rmi,L_j^{\rmi}}^{(j)} \}$ $(L_j^{\rmi}=|{\bm y}_{\rmi}^{(j)}|)$, ${\bm x}^{(j+1)}$, ${\bm y}_\rmii^{(j)}=\{
y_{\rmii,1}^{(j)},\dots,y_{\rmii,L_j^{\rmii}}^{(j)} \}$ ($L_j^{\rmii}=|{\bm y}_{\rmii}^{(j)}|$).
Note that the ordering of variables matters for quantum spaces in each layer,
and the variables are ordered as $y_{\rmi,1}^{(j)},\dots,y_{\rmi,L_j^{\rmi}}^{(j)}, x^{(j+1)}_1,\dots,x^{(j+1)}_{k_j}, y_{\rmii,1}^{(j)},\dots,y_{\rmii,L_j^{\rmii}}^{(j)}$
from left to right in the $j$-th layer ($j=1,\dots,M-1$).
The exception is the $M$-th layer, where the set of spectral variables of the quantum spaces is ${\bm y}^{(M)}=\{y^{(M)}_1,\dots,y^{(M)}_{L_M} \}$
$(L_M=|{\bm y}^{(M)}|)$.

To each quantum space with spectral variable in
${\bm y}_{\rmi}^{(j)}$, ${\bm y}_{\rmii}^{(j)}$, $(j=1,\dots,M-1)$, ${\bm y}^{(M)}$, we assign a color, which represents contraction by the corresponding basis vector. 
In Figure \ref{fig:rationalpartitionfunctions}, this is denoted by a circle 
on top of the vertical line (quantum space).
To label this assignment of colors, we introduce \emph{coordinates}.
We order the variables as ${\bm y}_{\rmi}^{(1)},{\bm y}_{\rmi}^{(2)},\cdots,{\bm y}_{\rmi}^{(M-1)},{\bm y}^{(M)}, {\bm y}_{\rmii}^{(M-1)},\dots,
{\bm y}_{\rmii}^{(2)},{\bm y}_{\rmii}^{(1)}$,
and call the quantum space to
which the $j$-th spectral variable is associated
the \emph{$j$-th quantum space}.
We call the place to which a color is assigned in the $j$-th quantum space the \emph{$j$-th coordinate}, 
running from $1, \dots, L^{\rmi}_1 + \dots + L^{\rmi}_{M-1} + L_M + L^{\rmii}_{M-1} + \dots + L^{\rmii}_1$.
If the place is in the $k$-th layer, we denote the color assigned as $i_j^{(k)}$,
see Figure \ref{fig:rationalpartitionfunctions}.
Knowledge of the colors at each coordinate and the spectral parameters is enough to specify the partition function. 
However, for the purposes of what follows, there is a more convenient way of parametrising the assignment of colors, which was introduced by Foda-Manabe.
Accordingly, we will refer to the above description as the \emph{naive labelling} of configurations, and we proceed to introduce the \emph{Foda-Manabe labelling} below.

Let $\widehat{{\bm I}}_{k_j}^{(j)}$ denote the ordered set of all coordinates from the $j$-th to the $M$-th layer, so
\begin{equation}
    \widehat{{\bm I}}_{k_j}^{(j)} = \left\{
        L^{\rmi}_1 + \cdots + L^{\rmi}_{j-1} + 1 < \cdots < L^{\rmi}_1 + \cdots + L^{\rmi}_{M-1} + L_M + L^{\rmii}_{M-1} + \cdots + L^{\rmii}_{j}
    \right\},
\end{equation}
with $L^\rmi_0 :=0$.
Let ${\bm I}_{k_j}^{(j)} \subset \widehat{{\bm I}}_{k_j}^{(j)}$ denote the coordinates among these which are colored by $1,2,\dots,j$.
Then $\widetilde{{\bm I}}_{k_j}^{(j)}$ is the set of coordinates which is \emph{induced} from ${\bm I}_{k_j}^{(j)}$ and $\widehat{{\bm I}}_{k_j}^{(j)}$, a process that is defined as follows.
We remove from $\widehat{{\bm I}}_{k_j}^{(j)}$ all coordinates colored by $j+2,\dots,M+1$ and relabel the remaining coordinates to $\{1,2,\dots,k_{j+1}+L_j^{\rmi}+L_j^{\rmii} \}$, preserving ordering.
Accordingly, the subset ${\bm I}_{k_j}^{(j)}$ of $\widehat{{\bm I}}_{k_j}^{(j)}$ is mapped
to some subset in $\{1,2,\dots,
k_{j+1}+L_j^{\rmi}+L_j^{\rmii} \}$; this is the definition of $\widetilde{{\bm I}}_{k_j}^{(j)}$.
To be more concrete,
see Figure \ref{inducedlabelone} for the understanding of the induced label
where $\{ I_1 < I_2 < \cdots < I_{k_{j+1}+L_j^{\rmi}+L_j^{\rmii}} \}$ is the set of
all coordinates colored by $1,2,\dots,j+1$ from the $j$-th to the $M$-th layer 
and is a subset of $\widehat{{\bm I}}_{k_j}^{(j)}$.
See also
Figure \ref{inducedlabeltwo}
for the graphical understanding of the sets introduced.
We denote the tuple $({\bm I}_{k_1}^{(1)}, {\bm I}_{k_2}^{(2)},\dots, {\bm I}_{k_M}^{(M)}   )$ by ${\bm I}$.
We mainly use this for the label of configurations.

Let us also give some explanation on the size of sets.
For example, note $|{\bm I}_{k_j}^{(j)}|=|\widetilde{{\bm I}}_{k_j}^{(j)}|=k_j$.
 This can be shown as follows. 
    The argument hinges on the fact that outputs of colour $j$ must occur in the layer $j-1$ or higher. 
    In the extreme, we can consider the top layer, and the set $\bm {I}^{(M)}_{k_M}$, which is the set of all coordinates in the top layer except those of colour $M+1$.  
    Comparing inputs and outputs of $M+1$, we have $L_M$ inputs and $k_M$ outputs in the auxiliary spaces. 
    Therefore, we must have $L_M-k_M$ outputs of $M+1$ in the top layer.
    The total number of top layer coordinates is $L_M$. 
    Hence, $|\bm {I}^{(M)}_{k_M}| = L_M - (L_M-k_M) = k_M$.
    The argument generalises to each $j$ similarly. 

    Another way of viewing the result is by starting from a position with no auxiliary spaces, in which case the input and output must be identical for each quantum space; in other words, all outputs at layer $j$ are of colour $j$.
    Then, for each auxiliary space that has output of colour $j$, the number of outputs that must be of colour $j$ is reduced by one. 
    The set of coordinates that are not colour $j$ is $\bm{I}^{(j)}_{k_j}$, and so it has size $k_j$.
We remark that ${\bm I}_{k_j}^{(j)}$ and $\widetilde{{\bm I}}_{k_j}^{(j)}$
are natural generalizations of the ones in Foda-Manabe whereas the symbols
$\widehat{{\bm I}}_{k_j}^{(j)}$ do not correspond to the ones used in their paper.
We also introduce notations for the elements of the sets $\bm I_{k_j}^{(j)}=\{ I_1^{(j)}< \cdots < I_{k_j}^{(j)} \}$.

\begin{figure}[ht] 
    \centering
        \includegraphics{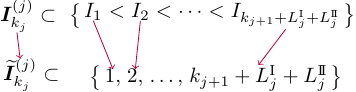}
\caption{Pictorial explanation of the induced label.}
\label{inducedlabelone}
    \end{figure}

\begin{figure}[ht]
        \centering
        \includegraphics{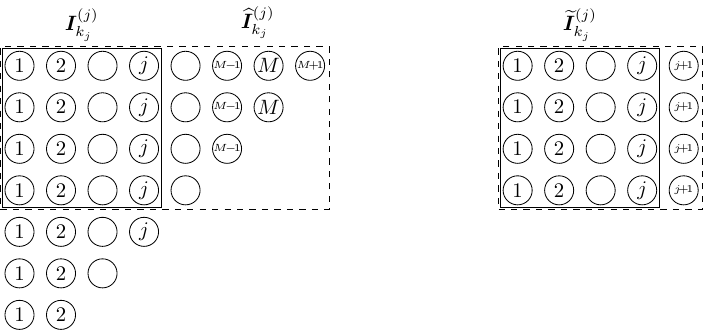}
\caption{Sets used to describe the configurations.
On the left panel, each row corresponds to each layer of the partition functions.
The bottom row corresponds to the first layer where only colors 1 and 2 are allowed.
The row above the bottom row corresponds to the second layer where color 3 is additionally allowed, and so on.
 ${\bm I}_{k_j}^{(j)}$ denote the set of coordinates which are colored by $1,2,\dots,j$ from the $j$-th to the $M$-th layer,
which is represented as the thick line.
The set of all coordinates  from the $j$-th to the $M$-th layer is denoted as $\widehat{{\bm I}}_{k_j}^{(j)}$
and corresponds to the dashed line.
The right panel represents the induced set $\widetilde{{\bm I}}_{k_j}^{(j)}$.
The set obtained by removing from $\widehat{{\bm I}}_{k_j}^{(j)}$ coordinates colored by $j+2,\dots,M+1$ and squeezing to $\{1,2,\dots,k_{j+1}+L_j^{\rmi}+L_j^{\rmii} \}$
is represented as the dashed line,
and the subset ${\bm I}_{k_j}^{(j)}$ is accordingly
mapped to some subset in $\{1,2,\dots,
k_{j+1}+L_j^{\rmi}+L_j^{\rmii} \}$ which is denoted as $\widetilde{{\bm I}}_{k_j}^{(j)}$ and corresponds to the thick line.
}
\label{inducedlabeltwo}
    \end{figure}

We denote the rational partition function as
  \begin{equation}
            \psi\Big(\bm x^{(1)}, \dots,\bm x^{(M)} \Big| \left\{ \bm y^{(1)}_{\rmi}, \bm y^{(1)}_{\rmii} \right\} , \dots, \left\{ \bm y^{(M-1)}_{\rmi}, \bm y^{(M-1)}_{\rmii} \right\} , \bm y^{(M)} \Big|
            k_1, \dots, k_M \Big| \{L_1^{\rmi}, L_1^{\rmii} \}, \dots, \{L_{M-1}^{\rmi}, L_{M-1}^{\rmii} \}, L_M \Big| \bm I \Big).
\label{rationalpartitionfunctions}
        \end{equation}
We will often abbreviate this to $\psi$ when we do not have to specify the detailed information of the partition functions.

The type of partition function we consider is a generalization of the one considered by Foda-Manabe.
The difference is the presence of additional ``left'' and ``right'' quantum sites at every level of the chain.
Setting ${\bm y}_{\rmi}^{(j)}=\emptyset$ $(L_j^{\rmi}=0)$ or ${\bm y}_{\rmii}^{(j)}=\emptyset$ $(L_j^{\rmii}=0)$ for all $1 \leq j \leq M-1$ reduces to the Foda-Manabe type. 
Note, however, that the top level of sites is unchanged; we may regard the left and right sites as being unified here.

\subsection{Characterization}
We determine the explicit form of the partition function
based on Lagrange interpolation, which we call the Izergin-Korepin method.
We first characterize the properties of partition functions
by a nested version of Korepin's lemma.
Korepin's lemma corresponds to constructing the recursion relations and initial conditions
of the partition functions.

\begin{proposition} \label{rationalKorepinlemma}
    Let $L^\rmi := L_1^\rmi + \cdots + L_{M-1}^\rmi + L_M$. 
    The partition function  $\psi$ satisfies the following properties: 
    \begin{enumerate}
        \item \emph{Case A}: If $i_{L^{\rmi}}^{(M)} \neq M+1$,
the degree of $y_{L_M}^{(M)}$ in $\psi$ is at most $k_M-1$.
        \item $\psi$ is symmetric with respect to $\bm x^{(M)}$.
        \item \emph{Case A}: If $i^{(M)}_{L^\rmi} \neq M+1$, 
        \begin{multline}
            \psi\Big(\bm x^{(1)}, \dots,\bm x^{(M)} \Big| \left\{ \bm y^{(1)}_{\rmi}, \bm y^{(1)}_{\rmii} \right\} , \dots, \left\{ \bm y^{(M-1)}_{\rmi}, \bm y^{(M-1)}_{\rmii} \right\} , \bm y^{(M)} \Big|
            \\
            k_1, \dots, k_M \Big| \{L_1^{\rmi}, L_1^{\rmii} \}, \dots, \{L_{M-1}^{\rmi}, L_{M-1}^{\rmii} \}, L_M \Big| \bm I \Big) \Bigg|_{y^{(M)}_{L_M}=x^{(M)}_{k_M}+1}
            \\
            = \prod_{j=1}^{L_M-1} \left(x^{(M)}_{k_M}-y^{(M)}_j\right)
            \prod_{j=1}^{k_M-1} \left(x^{(M)}_j-x^{(M)}_{k_M}-1\right)
            \\
            \times 
            \psi\Big(\bm x^{(1)}, \dots,\bm x^{(M)}  \setminus \{ x^{(M)}_{k_M} \} \Big| \left\{ \bm y^{(1)}_{\rmi}, \bm y^{(1)}_{\rmii} \right\} , \dots, \left\{ \bm y^{(M-1)}_{\rmi}, \{x^{(M)}_{k_M} \} \cup \bm y^{(M-1)}_{\rmii} \right\} , \bm y^{(M)} \setminus \{ y^{(M)}_{L_M} \} \Big|
            \\
            k_1, \dots, k_{M-1}, k_M-1 \Big| \{L_1^{\rmi}, L_1^{\rmii} \}, \dots, \{L_{M-1}^{\rmi}, L_{M-1}^{\rmii}+1 \}, L_M-1 \Big| \bm J \Big),
        \end{multline}
where the $\bm J$ are related to the $\bm I$ by 
        \begin{gather*}
            {\bm J}^{(i)}_{k_i} = {\bm I}^{(i)}_{k_i} \qquad 1 \leq i \leq M-1; \\
            {\bm J}^{(M)}_{k_M-1} = {\bm I}^{(M)}_{k_M} \setminus \{L^{\rmi} \}; \\
            \widetilde{{\bm J}}^{(i)}_{k_i} = \widetilde{{\bm I}}^{(i)}_{k_i} \qquad 1 \leq i \leq M-1.
        \end{gather*}
        \item \emph{Case B}: If $i^{(M)}_{L^{\rmi}} = M+1$,
        \begin{multline}
            \psi\Big(\bm x^{(1)}, \dots,\bm x^{(M)} \Big| \left\{ \bm y^{(1)}_{\rmi}, \bm y^{(1)}_{\rmii} \right\} , \dots, \left\{ \bm y^{(M-1)}_{\rmi}, \bm y^{(M-1)}_{\rmii} \right\} , \bm y^{(M)} \Big|
            \\
            k_1, \dots, k_M \Big| \{L_1^{\rmi}, L_1^{\rmii} \}, \dots, \{L_{M-1}^{\rmi}, L_{M-1}^{\rmii} \}, L_M \Big| \bm I \Big)
            \\
            =
            \prod_{j=1}^{k_M}\left(x_j - y^{(M)}_{L_M}+1 \right)
            \times 
            \psi\Big(\bm x^{(1)}, \dots,\bm x^{(M)} \Big| \left\{ \bm y^{(1)}_{\rmi}, \bm y^{(1)}_{\rmii} \right\} , \dots, \left\{ \bm y^{(M-1)}_{\rmi}, \bm y^{(M-1)}_{\rmii} \right\} , \bm y^{(M)} \setminus \{y^{(M)}_{L_M}\} \Big|
            \\
            k_1, \dots, k_M \Big| \{L_1^{\rmi}, L_1^{\rmii} \}, \dots, \{L_{M-1}^{\rmi}, L_{M-1}^{\rmii} \}, L_M-1 \Big| \bm J \Big),
        \end{multline}
 where the $\bm J$ are related to the $\bm I$ by 
        \begin{gather*}
            {\bm J}^{(i)}_{k_i} = \{j | j \in {\bm I}^{(i)}_{k_i}, j\leq L^{\rmi} \} \cup \{j-1 | j \in {\bm I}^{(i)}_{k_i}, j > L^{\rmi} \} \qquad 1 \leq i \leq M-1; \\
            {\bm J}^{(M)}_{k_M} = {\bm I}^{(M)}_{k_M} \\
            \widetilde{{\bm J}}^{(i)}_{k_i} = \widetilde{{\bm I}}^{(i)}_{k_i} \qquad 1 \leq i \leq M-1.
        \end{gather*}
        \item \emph{Initial condition}: If $k_M=1$ and $i^{(M)}_{L^{\rmi}} \neq M+1$,
            \begin{multline}
                \psi\Big(\bm x^{(1)}, \dots,\bm x^{(M)} \Big| \left\{ \bm y^{(1)}_{\rmi}, \bm y^{(1)}_{\rmii} \right\} , \dots, \left\{ \bm y^{(M-1)}_{\rmi}, \bm y^{(M-1)}_{\rmii} \right\} , \bm y^{(M)} \Big|
                \\
                k_1, \dots, k_M \Big| \{L_1^{\rmi}, L_1^{\rmii} \}, \dots, \{L_{M-1}^{\rmi}, L_{M-1}^{\rmii} \}, L_M \Big| \bm I \Big)
                \\
                =
                \prod_{j=1}^{L_M-1}\left(x^{(M)}_1-y^{(M)}_j\right) \times
                \\
                \times 
                \psi\Big(\bm x^{(1)}, \dots,\bm x^{(M-1)} \Big| \left\{ \bm y^{(1)}_{\rmi}, \bm y^{(1)}_{\rmii} \right\} , \dots, \left\{ \bm y^{(M-2)}_{\rmi}, \bm y^{(M-2)}_{\rmii} \right\} , \bm y^{(M-1)}_{\rmi} \cup \{x^{(M)}_1\} \cup \bm y^{(M-1)}_{\rmii} \Big|
                \\
                k_1, \dots, k_M \Big| \{L_1^{\rmi}, L_1^{\rmii} \}, \dots, \{L_{M-2}^{\rmi}, L_{M-2}^{\rmii} \}, L^{\rmi}_{M-1} + L^{\rmii}_{M-1} + 1 \Big| \bm J \Big),
            \end{multline}
where the $\bm J$ are related to the $\bm I$ by
    \begin{gather*}
        {{\bm J}}_{k_i}^{(i)} =\{ j| j \in {{\bm I}}_{k_i}^{(i)}, j<L^{\rmi}   \} \cup \{ j-L_M+1| j \in {{\bm I}}_{k_i}^{(i)}, j \geq L^{\rmi} \}  \qquad 1 \leq i \leq M-1;
\\
        \widetilde{{\bm J}}_{k_i}^{(i)} = \widetilde{{\bm I}}_{k_i}^{(i)} \qquad 1 \leq i \leq M-1.
    \end{gather*}
    \end{enumerate}
\end{proposition}

\begin{proof}
Property 1 can be checked by inserting the completeness relation between the rightmost column and the column left to it
of the top layer and using the definition of the $R$-matrix.

Property 2 follows from the standard train argument using the Yang-Baxter equation \eqref{rationalYBE}.

To prove Property 3, first note that in Case A, a number of nodes are ``frozen''
when $y^{(M)}_{L_M}$ is specialized to
$y^{(M)}_{L_M}=x^{(M)}_{k_M}+1$, as indicated in Figure \ref{frozenillustrationone}.
    \begin{figure*}[ht]
        \includegraphics[width=\textwidth]{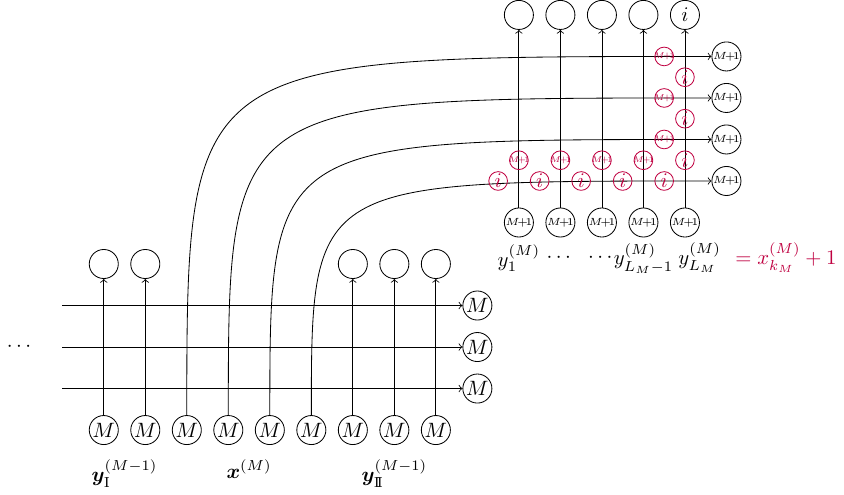}
        \caption{Pictorial explanation of property 3 of the partition function. At this value of $y^{(M)}_{L_M}$, the nodes are frozen in this configuration. }
\label{frozenillustrationone}
    \end{figure*}
These frozen nodes are evaluated as in Figure \ref{frozenweights}.
\begin{figure*}[ht]
    \centering
    \includegraphics{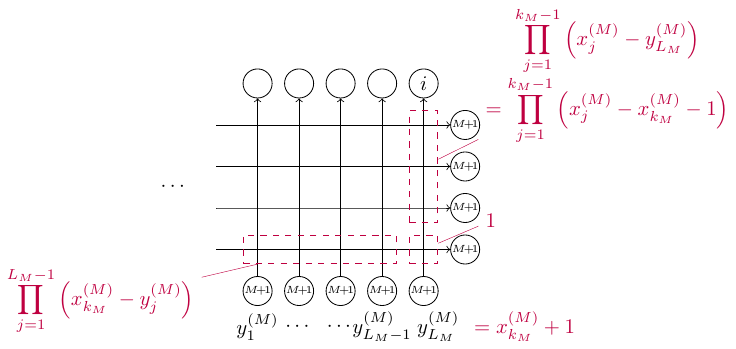}
\caption{Evaluation of the weights from the frozen nodes.}
\label{frozenweights}
\end{figure*}
This gives the recursion relation 
    \[
        \psi \left(\cdots \middle| {\bm I} \right)\Bigg|_{y^{(M)}_{L_M}=x^{(M)}_{k_M}+1}
        = 
 \prod_{j=1}^{L_M-1} \left(x^{(M)}_{k_M}-y^{(M)}_j\right)
            \prod_{j=1}^{k_M-1} \left(x^{(M)}_j-x^{(M)}_{k_M}-1\right)
 \psi \left(\cdots \middle| {\bm J} \right).
    \]

To prove Property 4, note that  In Case B, the rightmost column is frozen out
as Figure \ref{factorizationillustration}.
\begin{figure*}[ht]
    \centering
    \includegraphics{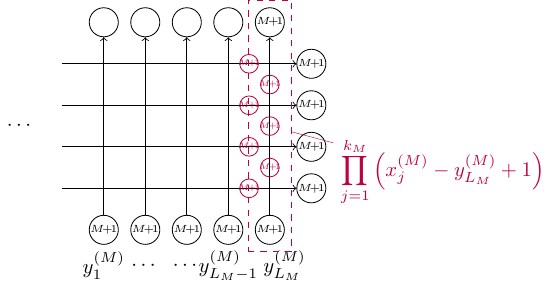}
\caption{Pictorial explanation for Case B.}
\label{factorizationillustration}
\end{figure*}
    This gives the recursion 
    \[
        \psi\left(\cdots \middle| {\bm I} \right) = \prod_{j=1}^{k_M} \left(x_j^{(M)}-y^{(M)}_{L_M}+1 \right) \, \psi \left(\cdots \middle| {\bm J} \right).
    \]

Finally let us check Property 5
corresponding to the initial condition for recursion: $k_M=1$ and $i^{(M)}_{L^{\rmi}} =j \neq M+1$.
See Figure \ref{initialconditionfigure}.
    Diagrammatically we can see that this freezes out the entire top layer to give a factor
    \[
        \prod_{i=1}^{L_M-1} \left(x_1^{(M)} - y_i^{(M)} \right).
    \]
Removing the top layer, we obtain partition functions of $\mathfrak{gl}_M$
labelled by $\bm J$, as given in Property 5.
    \begin{figure*}[ht]
        \centering
        \includegraphics[width=0.7\textwidth]{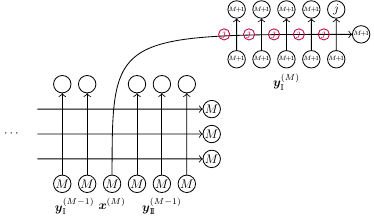}
\caption{The case when $k_M=1$ and $i^{(M)}_{L^{\rmi}} =j \neq M+1$.
The top layer is frozen, and removing that part gives partition function of $\mathfrak{gl}_M$.}
\label{initialconditionfigure}
    \end{figure*}

\end{proof}


Before going to the next subsection,
we explain the nested Izergin-Korepin method.
We use induction on $(M,L_M)$, ordered lexicographically, $k_M \leq L_M$. 
Suppose the explicit forms of the $\mathfrak{gl}_M$ partition functions are determined,
and we wish to determine the $\mathfrak{gl}_{M+1}$ partition functions.
If $k_M=1$, we may apply Property 5 to reduce the expression to a known $\mathfrak{gl}_M$ partition function.
Otherwise, we apply a nested induction on $L_M$; the remaining properties may be used to reduce any 
partition function to one with $L_M$ replaced by $L_M-1$, and since $k_M \leq L_M$ at all times, we eventually arrive at $k_M=1$.
The inductive step is clear when $i_{L_1^{\rmi}+\cdots+L_{M-1}^{\rmi}+L_M}^{(M)} = M+1$, using Property 4. 
In the case $i_{L_1^{\rmi}+\cdots+L_{M-1}^{\rmi}+L_M}^{(M)} \neq M+1$, Property 1 implies that we need only evaluate $\psi$ at $k_M$ points to determine its $y_{L_M}^{(M)}$-dependence, Property 3 gives one such point as $y_{L_M}^{(M)} = x_{k_M}^{(M)} + 1$, and Property 2 allows us to replace $x_{k_M}^{(M)}$ by any other $x_{j}^{(M)}$, giving a total of $k_M$ points and completing the induction.

\subsection{Multisymmetric functions}
In this subsection, we introduce a class of multisymmetric functions
and show they are explicit forms of the partition functions for the rational case.
\begin{definition}
We define the extended rational weight function as
\begin{multline} \label{4-extended-wavefunction}
        W \bigg( \bm x^{(1)} ,\dots, \bm x^{(M)} \bigg| 
        \{\bm y_{\rmi}^{(1)},\bm y_{\rmii}^{(1)} \} ,\dots, \{\bm y_{\rmi}^{(M-1)},\bm y_{\rmii}^{(M-1)} \}, \bm y^{(M)} \bigg| 
        \\
        k_1, \dots, k_M \bigg| \{L_1^{\rmi} , L_1^{\rmii} \}, \dots, \{L_{M-1}^{\rmi} , L_{M-1}^{\rmii} \}, L_M \bigg|  \bm I \bigg)
        \\
        = \sum_{\sigma_1 \in S_{k_1}} \cdots \sum_{\sigma_M \in S_{k_M}}
        \prod_{p=1}^{M-1} \Bigg\{
            \prod_{a=1}^{k_p} \Bigg(
                \prod_{i=1}^{\widetilde{I}_{a}^{(p)}-1-L_1^{\rmi}-\cdots -L_{p-1}^{\rmi}} 
                    \left(x^{(p)}_{\sigma_p(a)}-m_i^{L_p^{\rmi},k_{p+1},L_p^{\rmii}} \left(\bm y_{\rmi}^{(p)}\middle| \bm x^{(p+1)}_{\sigma_{p+1}} \middle| \bm y^{(p)}_{\rmii}\right) \right)
                \\
                \times 
                \prod_{i=\widetilde{I}_{a}^{(p)}+1-L_1^{\rmi}-\cdots -L_{p-1}^{\rmi}}^{L^{\rmi}_p+L^{\rmii}_p + k_{p+1}}
                    \left(x^{(p)}_{\sigma_p(a)}-m_i^{L_p^{\rmi},k_{p+1},L_p^{\rmii}} \left(\bm y_{\rmi}^{(p)}\middle| \bm x^{(p+1)}_{\sigma_{p+1}} \middle| \bm y^{(p)}_{\rmii}\right)+1 \right)
            \Bigg)
            \prod_{a<b}^{k_p} 
                \frac{x^{(p)}_{\sigma_p(a)}-x^{(p)}_{\sigma_p(b)}-1}{x^{(p)}_{\sigma_p(a)}-x^{(p)}_{\sigma_p(b)}}
        \Bigg\}
        \\
        \times 
        \prod_{a=1}^{k_M} \left(
            \prod_{i=1}^{I_a^{(M)}-1-L_1^{\rmi}-\cdots - L_{M-1}^{\rmi}} 
                \left(x_{\sigma_M(a)}^{(M)}-y_i^{(M)}\right)
            \prod_{i=I_a^{(M)}+1-L_1^{\rmi}-\cdots - L_{M-1}^{\rmi}}^{L_M} 
                \left(x_{\sigma_M(a)}^{(M)}-y_i^{(M)}+1\right)
        \right)
        \\ \times 
        \prod_{a<b}^{k_M} 
            \frac{x^{(M)}_{\sigma_M(a)}-x^{(M)}_{\sigma_M(b)}-1}{x^{(M)}_{\sigma_M(a)}-x^{(M)}_{\sigma_M(b)}},
    \end{multline}
    where
    \begin{equation*}
        m_i^{L_p^{\rmi},k_{p+1},L_p^{\rmii}} \left(\bm y_{\rmi}^{(p)}\middle| \bm x^{(p+1)} \middle| \bm y^{(p)}_{\rmii}\right) = 
        \begin{cases}
            y^{(p)}_{\rmi,i} & 1 \leq i \leq L_p^{\rmi} \\
            x^{(p+1)}_{i-L_p^{\rmi}} & L_p^{\rmi}+1 \leq i \leq L_p^{\rmi}+k_{p+1} \\
            y^{(p)}_{\rmii,i-L_p^{\rmi}-k_{p+1}} & L_p^{\rmi}+k_{p+1}+1 \leq i \leq L_p^{\rmi}+k_{p+1}+L_p^{\rmii}
        \end{cases},
    \end{equation*}
    and ${\bm x}^{(j)}_{\sigma_j} := \left\{ x^{(j)}_{\sigma_j(1)}, \dots, x^{(j)}_{\sigma_j(k_j)}\right\}$ as an ordered set. 
\end{definition}

\begin{theorem}
We have
\begin{multline}
 \psi \bigg( \bm x^{(1)} ,\dots, \bm x^{(M)} \bigg| 
        \{\bm y_{\rmi}^{(1)},\bm y_{\rmii}^{(1)} \} ,\dots, \{\bm y_{\rmi}^{(M-1)},\bm y_{\rmii}^{(M-1)} \}, \bm y^{(M)} \bigg| \\
        k_1, \dots, k_M \bigg| \{L_1^{\rmi} , L_1^{\rmii} \}, \dots, \{L_{M-1}^{\rmi} , L_{M-1}^{\rmii} \}, L_M \bigg|  \bm I \bigg)
        \\
        =
 W \bigg( \bm x^{(1)} ,\dots, \bm x^{(M)} \bigg| 
        \{\bm y_{\rmi}^{(1)},\bm y_{\rmii}^{(1)} \} ,\dots, \{\bm y_{\rmi}^{(M-1)},\bm y_{\rmii}^{(M-1)} \}, \bm y^{(M)} \bigg| \\
        k_1, \dots, k_M \bigg| \{L_1^{\rmi} , L_1^{\rmii} \}, \dots, \{L_{M-1}^{\rmi} , L_{M-1}^{\rmii} \}, L_M \bigg|  \bm I \bigg).
\end{multline}
\end{theorem}

\begin{proof}
What we need to prove is that the function \eqref{4-extended-wavefunction}
satisfies all the properties listed in Proposition \ref{rationalKorepinlemma}, as these uniquely determine the partition function.
Properties 1 and 2 are easy to check from the definition, while the confirmations of Properties 3, 4 and 5 are given as propositions below.
\end{proof}

 \begin{proposition}
If $i_{L^{\rmi}} \neq M+1$,
the functions \eqref{4-extended-wavefunction} satisfy the recursion relation 
        \begin{multline} \label{4-caseA}
            W \bigg( \bm x^{(1)} ,\dots, \bm x^{(M)} \bigg| 
            \{\bm y_{\rmi}^{(1)},\bm y_{\rmii}^{(1)} \} ,\dots, \{\bm y_{\rmi}^{(M-1)},\bm y_{\rmii}^{(M-1)} \}, \bm y^{(M)} \bigg| \\
            k_1, \dots, k_M \bigg| \{L_1^{\rmi} , L_1^{\rmii} \}, \dots, \{L_{M-1}^{\rmi} , L_{M-1}^{\rmii} \}, L_M \bigg|  \bm I \bigg) \Bigg|_{y^{(M)}_{L_M}=x^{(M)}_{k_M}+1}
            \\
            =
            \prod_{j=1}^{L_M-1} \left(x^{(M)}_{k_M}-y^{(M)}_j\right)
            \prod_{j=1}^{k_M-1} \left(x^{(M)}_{j}-x^{(M)}_{k_M}-1\right)
            \\
            \times 
           W \bigg( \bm x^{(1)} ,\dots, \bm x^{(M)} \setminus \{x^{(M)}_{k_M}\}\bigg| 
            \{\bm y_{\rmi}^{(1)},\bm y_{\rmii}^{(1)} \} ,\dots, \{\bm y_{\rmi}^{(M-1)},\{x^{(M)}_{k_M}\}\cup \bm y_{\rmii}^{(M-1)} \}, \bm y^{(M)} \setminus \{y^{(M)}_{L_M}\} \bigg| 
            \\
            k_1, \dots, k_{M-1}, k_M-1 \bigg| \{L_1^{\rmi} , L_1^{\rmii} \}, \dots, \{L_{M-1}^{\rmi} , L_{M-1}^{\rmii}+1 \}, L_M-1 \bigg|  \bm J \bigg),
        \end{multline}
        where the $\bm J$ are related to the $\bm I$ by 
        \begin{gather*}
           {\bm J}^{(i)}_{k_i} ={\bm I}^{(i)}_{k_i} \qquad 1 \leq i \leq M-1; \\
           {\bm J}^{(M)}_{k_M-1} ={\bm I}^{(M)}_{k_M} \setminus \{L^{\rmi}\}; \\
            \widetilde{{\bm J}}^{(i)}_{k_i} = \widetilde{{\bm I}}^{(i)}_{k_i} \qquad 1 \leq i \leq M-1.
        \end{gather*}
    \end{proposition}
    
    \begin{proof}
        We check by the following two steps.

        \emph{Step 1. -- recursion of the ``final factor''.} 
        We will refer to the following as the ``final factor'' of the weight functions.
        \begin{equation}\label{4:caseA-final-factor}
            \underbrace{\prod_{a=1}^{k_M} 
                \prod_{i=1}^{I_a^{(M)}-1 -\sum_{q=1}^{M-1}L^{\rmi}_q} 
                    \left(x_{\sigma_M(a)}^{(M)}-y_i^{(M)}\right)}_{X}
            \underbrace{
            \prod_{a=1}^{k_M} 
                \prod_{i=I_a^{(M)}+1-\sum_{q=1}^{M-1}L^{\rmi}_q}^{L_M} 
                    \left(x_{\sigma_M(a)}^{(M)}-y_i^{(M)}+1\right)}_{Y}
            \underbrace{
            \prod_{a<b}^{k_M} 
                \frac{x^{(M)}_{\sigma_M(a)}-x^{(M)}_{\sigma_M(b)}-1}{x^{(M)}_{\sigma_M(a)}-x^{(M)}_{\sigma_M(b)}}}_{Z}.
        \end{equation}
        The expression above is made up of three products, which we will analyse in turn.
        First, note that for Case A, we have $I^{(M)}_{k_M} = L_1^{\rmi}+\cdots +L_{M-1}^{\rmi}+L_M$, so the second product of $Y$ is empty when $a=k_M$:
        \begin{equation*}
            Y =
            \prod_{a=1}^{k_M-1} 
                \prod_{i=I_a^{(M)}+1-\sum_{q=1}^{M-1}L^{\rmi}_q}^{L_M} 
                    \left(x_{\sigma_M(a)}^{(M)}-y_i^{(M)}+1\right)
        \end{equation*}
        Specialising to $y^{(M)}_{L_M}=x^{(M)}_{k_M}+1$, 
        \begin{align*}
            Y &= \prod_{a=1}^{k_M-1} 
            \prod_{i=I_a^{(M)}+1 -\sum_{q=1}^{M-1}L^{\rmi}_q}^{L_M-1} 
                \left(x_{\sigma_M(a)}^{(M)}-y_i^{(M)}+1\right)
            \prod_{a=1}^{k_M-1} 
            \left(x_{\sigma_M(a)}^{(M)}-y_{L_M}^{(M)}+1\right)
            \\
            &=\prod_{a=1}^{k_M-1} 
                \prod_{i=I_a^{(M)}+1 -\sum_{q=1}^{M-1}L^{\rmi}_q}^{L_M-1} 
                    \left(x_{\sigma_M(a)}^{(M)}-y_i^{(M)}+1\right)
            \prod_{a=1}^{k_M-1} 
                \left(x_{\sigma_M(a)}^{(M)}-x_{k_M}^{(M)}\right),
        \end{align*}
        we see that this product is zero unless $\sigma_{M}(k_M) = k_M$. 
        We rearrange the final product to arrive at
        \[
            Y = \prod_{a=1}^{k_M-1} 
                \prod_{i=I_a^{(M)}+1 -\sum_{q=1}^{M-1}L^{\rmi}_q}^{L_M-1} 
                    \left(x_{\sigma_M(a)}^{(M)}-y_i^{(M)}+1\right)
            \prod_{a=1}^{k_M-1} 
                \left(x_{a}^{(M)}-x_{k_M}^{(M)}\right).
        \]
        Next, we focus on $Z$, which can be split as follows:
        \begin{align*}
            Z&=
            \prod_{a<b}^{k_M-1} 
                \frac{x^{(M)}_{\sigma_M(a)}-x^{(M)}_{\sigma_M(b)}-1}{x^{(M)}_{\sigma_M(a)}-x^{(M)}_{\sigma_M(b)}}
            \prod_{a=1}^{k_M-1} 
                \frac{x^{(M)}_{\sigma_M(a)}-x^{(M)}_{\sigma_M(k_M)}-1}{x^{(M)}_{\sigma_M(a)}-x^{(M)}_{\sigma_M(k_M)}}
                \\
                &=
            \prod_{a<b}^{k_M-1} 
                \frac{x^{(M)}_{\sigma_M(a)}-x^{(M)}_{\sigma_M(b)}-1}{x^{(M)}_{\sigma_M(a)}-x^{(M)}_{\sigma_M(b)}}
            \prod_{a=1}^{k_M-1} 
                \frac{x^{(M)}_{\sigma_M(a)}-x^{(M)}_{k_M}-1}{x^{(M)}_{\sigma_M(a)}-x^{(M)}_{k_M}},
        \end{align*}
        where we have used the fact that $\sigma_{M}(k_M) = k_M$.
        We can then reorder the second product here to give 
        \[
            Z=
            \prod_{a<b}^{k_M-1} 
                \frac{x^{(M)}_{\sigma_M(a)}-x^{(M)}_{\sigma_M(b)}-1}{x^{(M)}_{\sigma_M(a)}-x^{(M)}_{\sigma_M(b)}}
            \prod_{a=1}^{k_M-1} 
                \frac{x^{(M)}_{a}-x^{(M)}_{k_M}-1}{x^{(M)}_{a}-x^{(M)}_{k_M}}.
        \]
        Finally, we turn to $X$, which we split as 
        \begin{align*}
            X &=
            \prod_{a=1}^{k_M-1} 
                \prod_{i=1}^{I_a^{(M)} -1 -\sum_{q=1}^{M-1}L^{\rmi}_q} 
                    \left(x_{\sigma_M(a)}^{(M)}-y_i^{(M)}\right)
            \prod_{i=1}^{L_M-1} 
                \left(x_{\sigma_M(k_M)}^{(M)}-y_i^{(M)}\right).
                \\
                &=
            \prod_{a=1}^{k_M-1} 
                \prod_{i=1}^{I_a^{(M)} -1-\sum_{q=1}^{M-1}L^{\rmi}_q} 
                    \left(x_{\sigma_M(a)}^{(M)}-y_i^{(M)}\right)
            \prod_{i=1}^{L_M-1} 
                \left(x_{k_M}^{(M)}-y_i^{(M)}\right).
        \end{align*}
        Multiplying together the three factors, we obtain 
        \begin{multline*}
            \prod_{a=1}^{k_M-1} 
                \prod_{i=1}^{I_a^{(M)} -1-\sum_{q=1}^{M-1}L^{\rmi}_q} 
                    \left(x_{\sigma_M(a)}^{(M)}-y_i^{(M)}\right)
            \prod_{i=1}^{L_M-1} 
                \left(x_{k_M}^{(M)}-y_i^{(M)}\right)
                \\ \times
            \prod_{a=1}^{k_M-1} 
                \prod_{i=I_a^{(M)}+1 -\sum_{q=1}^{M-1}L^{\rmi}_q}^{L_M-1} 
                    \left(x_{\sigma_M(a)}^{(M)}-y_i^{(M)}+1\right)
            \cancel{\prod_{a=1}^{k_M-1} 
                \left(x_{a}^{(M)}-x_{k_M}^{(M)}\right)}
                \\ \times 
            \prod_{a<b}^{k_M-1} 
                \frac{x^{(M)}_{\sigma_M(a)}-x^{(M)}_{\sigma_M(b)}-1}{x^{(M)}_{\sigma_M(a)}-x^{(M)}_{\sigma_M(b)}}
            \prod_{a=1}^{k_M-1} 
                \frac{x^{(M)}_{a}-x^{(M)}_{k_M}-1}{\cancel{x^{(M)}_{a}-x^{(M)}_{k_M}}}.
        \end{multline*}
        This is equal to 
        \begin{multline*}
            \prod_{i=1}^{L_M-1} 
                \left(x_{k_M}^{(M)}-y_i^{(M)}\right)
            \prod_{a=1}^{k_M-1} 
                \left(x^{(M)}_{a}-x^{(M)}_{k_M}-1\right)
                \\ \times 
            \prod_{a=1}^{k_M-1} 
                \prod_{i=1}^{I_a^{(M)} -1-\sum_{q=1}^{M-1}L^{\rmi}_q} 
                    \left(x_{\sigma_M(a)}^{(M)}-y_i^{(M)}\right)
            \prod_{a=1}^{k_M-1} 
                \prod_{i=I_a^{(M)}+1 -\sum_{q=1}^{M-1}L^{\rmi}_q}^{L_M-1} 
                    \left(x_{\sigma_M(a)}^{(M)}-y_i^{(M)}+1\right)
                    \\ \times 
            \prod_{a<b}^{k_M-1} 
                \frac{x^{(M)}_{\sigma_M(a)}-x^{(M)}_{\sigma_M(b)}-1}{x^{(M)}_{\sigma_M(a)}-x^{(M)}_{\sigma_M(b)}}
        \end{multline*}
        We observe that this is equal to the residual factor from \eqref{4-caseA}, along with the ``final factor'' of the wavefunction on the right hand side of \eqref{4-caseA}.
        
        \emph{Step 2. -- recursion of the remaining factors.}
        In the second to last factor, 
        \begin{align*}
            &m_i^{L_{M-1}^{\rmi},k_{M},L_{M-1}^{\rmii}} \left(\bm y_{\rmi}^{(M-1)}\middle| \bm x^{(M)}_{\sigma_M} \middle| {\bm y}^{(M-1)}_{\rmii}\right) 
            \\
            &\qquad = 
            \begin{cases}
                y^{(M-1)}_{\rmi,i} & 1 \leq i \leq L_{M-1}^{\rmi} \\
                x^{(M)}_{\sigma_M(i-L_{M-1}^{\rmi})} & L_{M-1}^{\rmi}+1 \leq i \leq L_{M-1}^{\rmi}+k_{M} 
                \\
                y^{(M-1)}_{\rmii,i-L_{M-1}^{\rmi}-k_{M}} & L_{M-1}^{\rmi}+k_{M}+1 \leq i \leq L_{M-1}^{\rmi}+k_{M}+L_{M-1}^{\rmii}
            \end{cases}
            \\
            &\qquad = 
            \begin{cases}
                y^{(M-1)}_{\rmi,i} & 1 \leq i \leq L_{M-1}^{\rmi} \\
                x^{(M)}_{\sigma_M(i-L_{M-1}^{\rmi})} & L_{M-1}^{\rmi}+1 \leq i \leq L_{M-1}^{\rmi}+k_{M}-1 \\
                x^{(M)}_{k_M} & i = L_{M-1}^{\rmi}+k_{M} \\
                y^{(M-1)}_{\rmii,i-L_{M-1}^{\rmi}-k_{M}} & L_{M-1}^{\rmi}+k_{M}+1 \leq i \leq L_{M-1}^{\rmi}+k_{M}+L_{M-1}^{\rmii},
            \end{cases}
        \end{align*}
        where we have used $\sigma_M(k_M) = k_M$, which was established in Step 1. 
        Hence, this is equal to the symbol
        \[
            m_i^{L_{M-1}^{\rmi},k_{M}-1,L_{M-1}^{\rmii}+1} \left(\bm y_{\rmi}^{(M-1)}\middle| \bm x^{(M)}_{\sigma_M}\setminus \{x^{(M)}_{k_M}\} \middle| \{x^{(M)}_{k_M}\} \cup {\bm y}^{(M-1)}_{\rmii}\right) .
        \]
        Finally, we change all instances of $\widetilde{I}^{(M-1)}_a$ to $\widetilde{J}^{(M-1)}_a$ to prove the recursion.
    \end{proof}

    \begin{proposition}
If $i_{L^{\rmi}} = M+1$,
        the function \eqref{4-extended-wavefunction} satisfies the recursion relation 
        \begin{multline} \label{4-caseB}
           W \bigg( \bm x^{(1)} ,\dots, \bm x^{(M)} \bigg| 
            \{\bm y_{\rmi}^{(1)},\bm y_{\rmii}^{(1)} \} ,\dots, \{\bm y_{\rmi}^{(M-1)},\bm y_{\rmii}^{(M-1)} \}, \bm y^{(M)} \bigg| 
            k_1, \dots, k_M \bigg| \{L_1^{\rmi} , L_1^{\rmii} \}, \dots, \{L_{M-1}^{\rmi} , L_{M-1}^{\rmii} \}, L_M \bigg|  \bm I \bigg) 
            \\
            =
            \prod_{j=1}^{k_M}\left(x_j - y^{(M)}_{L_M}+1 \right)
            \times 
            W\Big(\bm x^{(1)}, \dots,\bm x^{(M)} \Big| \left\{ \bm y^{(1)}_{\rmi}, \bm y^{(1)}_{\rmii} \right\} , \dots, \left\{ \bm y^{(M-1)}_{\rmi}, \bm y^{(M-1)}_{\rmii} \right\} , \bm y^{(M)} \setminus \{y^{(M)}_{L_M}\} \Big|
            \\
            k_1, \dots, k_M \Big| \{L_1^{\rmi}, L_1^{\rmii} \}, \dots, \{L_{M-1}^{\rmi}, L_{M-1}^{\rmii} \}, L_M-1 \Big| \bm J \Big),
        \end{multline}
        where the $\bm J$ are related to the $\bm I$ by 
        \begin{gather*}
            {\bm J}^{(i)}_{k_i} = \{j | j \in {\bm I}^{(i)}_{k_i}, j\leq L^{\rmi} \} \cup \{j-1 | j \in {\bm I}^{(i)}_{k_i}, j > L^{\rmi} \} \qquad 1 \leq i \leq M-1; \\
            {\bm J}^{(M)}_{k_M} = {\bm I}^{(M)}_{k_M} ;
            \\
            \widetilde{{\bm J}}^{(i)}_{k_i} = \widetilde{{\bm I}}^{(i)}_{k_i} \qquad 1 \leq i \leq M-1.
        \end{gather*}
    \end{proposition}

    \begin{proof}
        Consider the ``final factor'' again:
        \begin{equation*}
            \prod_{a=1}^{k_M}
                \prod_{i=1}^{I_a^{(M)}-1-L_1^{\rmi}-\cdots - L_{M-1}^{\rmi}} 
                    \left(x_{\sigma_M(a)}^{(M)}-y_i^{(M)}\right)
            \prod_{a=1}^{k_M}
                \prod_{i=I_a^{(M)}+1-L_1^{\rmi}-\cdots - L_{M-1}^{\rmi}}^{L_M} 
                    \left(x_{\sigma_M(a)}^{(M)}-y_i^{(M)}+1\right)
            \prod_{a<b}^{k_M}
                \frac{x^{(M)}_{\sigma_M(a)}-x^{(M)}_{\sigma_M(b)}-1}{x^{(M)}_{\sigma_M(a)}-x^{(M)}_{\sigma_M(b)}}.
        \end{equation*}
        Of these three products, we focus on the middle product. 
        We have $I^{(M)}_{k_M} < L_1^{\rmi} + \cdots + L_{M_1}^{\rmi} + L_M$, so the product is non-empty. 
        Taking out the factor corresponding to $i=L_M$, we obtain
        \begin{multline*}
            \prod_{a=1}^{k_M} \left(x_{a}^{(M)}-y_{L_M}^{(M)}+1\right)
            \\ \times 
            \prod_{a=1}^{k_M} 
                \prod_{i=1}^{I_a^{(M)}-1-L_1^{\rmi}-\cdots - L_{M-1}^{\rmi}} 
                    \left(x_{\sigma_M(a)}^{(M)}-y_i^{(M)}\right)
            \prod_{a=1}^{k_M} 
                \prod_{i=I_a^{(M)}+1-L_1^{\rmi}-\cdots - L_{M-1}^{\rmi}}^{L_M-1} 
                    \left(x_{\sigma_M(a)}^{(M)}-y_i^{(M)}+1\right)
            \\ \times 
            \prod_{a<b}^{k_M} 
                \frac{x^{(M)}_{\sigma_M(a)}-x^{(M)}_{\sigma_M(b)}-1}{x^{(M)}_{\sigma_M(a)}-x^{(M)}_{\sigma_M(b)}}.
        \end{multline*}
        We see that this gives the recursion in \eqref{4-caseB}.
    \end{proof}

\begin{proposition}
If $k_M=1$ and $i_{L^{\rmi}} \neq M+1$,
        the functions \eqref{4-extended-wavefunction} satisfy
        \begin{multline*}
           W \bigg( \bm x^{(1)} ,\dots, \bm x^{(M)} \bigg| 
            \{\bm y_{\rmi}^{(1)},\bm y_{\rmii}^{(1)} \} ,\dots, \{\bm y_{\rmi}^{(M-1)},\bm y_{\rmii}^{(M-1)} \}, \bm y^{(M)} \bigg| \\
            k_1, \dots, k_{M-1}, 1 \bigg| \{L_1^{\rmi} , L_1^{\rmii} \}, \dots, \{L_{M-1}^{\rmi} , L_{M-1}^{\rmii} \}, L_M \bigg|  \bm I \bigg) 
            \\
            =  \left(\prod_{j=1}^{L_M-1}\left(x^{(M)}_1-y^{(M)}_j\right) \right) \times 
            \\
                W\Big(\bm x^{(1)}, \dots,\bm x^{(M-1)} \Big| \left\{ \bm y^{(1)}_{\rmi}, \bm y^{(1)}_{\rmii} \right\} , \dots, \left\{ \bm y^{(M-2)}_{\rmi}, \bm y^{(M-2)}_{\rmii} \right\} , \bm y^{(M-1)}_{\rmi} \cup \{x^{(M)}_1\} \cup \bm y^{(M-1)}_{\rmii} \Big|
                \\
                k_1, \dots, k_M \Big| \{L_1^{\rmi}, L_1^{\rmii} \}, \dots, \{L_{M-2}^{\rmi}, L_{M-2}^{\rmii} \}, L^{\rmi}_{M-1} + L^{\rmii}_{M-1} + 1 \Big| \bm J \Big),
        \end{multline*}
where the $\bm J$ are related to the $\bm I$ by
    \begin{gather*}
        {{\bm J}}_{k_i}^{(i)} =\{ j| j \in {{\bm I}}_{k_i}^{(i)}, j<L^{\rmi}   \} \cup \{ j-L_M+1| j \in {{\bm I}}_{k_i}^{(i)}, j \geq L^{\rmi} \}  \qquad 1 \leq i \leq M-1;
\\
        \widetilde{{\bm J}}_{k_i}^{(i)} = \widetilde{{\bm I}}_{k_i}^{(i)} \qquad 1 \leq i \leq M-1.
    \end{gather*}
    \end{proposition}
    
    \begin{proof}
        Consider once again the ``final factor'',
        \begin{equation*}
            \prod_{a=1}^{k_M}
                \prod_{i=1}^{I_a^{(M)}-1-L_1^{\rmi}-\cdots - L_{M-1}^{\rmi}} 
                    \left(x_{\sigma_M(a)}^{(M)}-y_i^{(M)}\right)
            \prod_{a=1}^{k_M}
                \prod_{i=I_a^{(M)}+1-L_1^{\rmi}-\cdots - L_{M-1}^{\rmi}}^{L_M} 
                    \left(x_{\sigma_M(a)}^{(M)}-y_i^{(M)}+1\right)
            \prod_{a<b}^{k_M}
                \frac{x^{(M)}_{\sigma_M(a)}-x^{(M)}_{\sigma_M(b)}-1}{x^{(M)}_{\sigma_M(a)}-x^{(M)}_{\sigma_M(b)}}.
        \end{equation*}
        Noting that $k_M=1$ and $\bm I^{(M)}_1= \{ L^{\rmi}\}$, we see that the factor is greatly reduced to
        \[
            \prod_{i=1}^{L_M-1} 
                    \left(x_{\sigma_M(1)}^{(M)}-y_i^{(M)}\right).
        \]
        For the second to last factor of the wavefunction, corresponding to level $M-1$, we use notation 
        \begin{align*}
            \widetilde{I}^{(M-1)}_a &= J_a^{(M-1)} \\
            \widetilde{L}_{M-1} &= 1 + L^{\rmi}_{M-1} + L^{\rmii}_{M-1} \\
            \widetilde{\bm y}^{(M-1)} &= \bm y_{\rmi}^{(M-1)} \cup \{ x_1^{(M)} \} \cup \bm y_{\rmii}^{(M-1)}.
        \end{align*}
        The second to last factor assumes the form
        \begin{multline*} 
            \prod_{a=1}^{k_{M-1}} \Bigg(
                \prod_{i=1}^{\widetilde{I}_{a}^{(M-1)}-1-L_1^{\rmi}-\cdots -L_{M-2}^{\rmi}} 
                    \left(x^{(M-1)}_{\sigma_{M-1}(a)}-m_i^{L_{M-1}^{\rmi},1,L_{M-1}^{\rmii}} \left(\bm y_{\rmi}^{(M-1)}\middle| \{ x_1^{(M)} \} \middle| \bm y^{(M-1)}_{\rmii}\right) \right)
                \\
                \times 
                \prod_{i=\widetilde{I}_{a}^{(M-1)}+1-L_1^{\rmi}-\cdots -L_{M-2}^{\rmi}}^{L^{\rmi}_{M-1}+L^{\rmii}_{M-1}+1}
                    \left(x^{(M-1)}_{\sigma_{M-1}(a)}-m_i^{L_{M-1}^{\rmi},1,L_{M-1}^{\rmii}} \left(\bm y_{\rmi}^{(M-1)} \middle|  \{ x_1^{(M)} \} \middle| \bm y^{(M-1)}_{\rmii}\right)+1 \right)
            \Bigg)
            \\
            \times
            \prod_{a<b}^{k_{M-1}} 
                \frac{x^{(M-1)}_{\sigma_{M-1}(a)}-x^{(M-1)}_{\sigma_{M-1}(b)}-1}{x^{(M-1)}_{\sigma_{M-1}(a)}-x^{(M-1)}_{\sigma_{M-1}(b)}},
        \end{multline*}
        where we see that
        \begin{align*}
            m_i^{L_{M-1}^{\rmi},1,L_{M-1}^{\rmii}} \left(\bm y_{\rmi}^{(M-1)}\middle| \{ x_1^{(M)} \} \middle| \bm y^{(M-1)}_{\rmii}\right) 
            &= 
            \begin{cases}
                y^{(M-1)}_{\rmi,i} & 1 \leq i \leq L_{M-1}^{\rmi} \\
                x^{(M)}_{i-L_{M-1}^{\rmi}} &  i = L_{M-1}^{\rmi}+1 \\
                y^{(M-1)}_{\rmii,i-L_{M-1}^{\rmi}-1} & L_{M-1}^{\rmi}+2 \leq i \leq L_{M-1}^{\rmi}+1+L_{M-1}^{\rmii}
            \end{cases}
            \\
            &= \widetilde{y}^{(M-1)}_{i}.
        \end{align*}

        Therefore, the second to last factor becomes
        \begin{multline*} 
            \prod_{a=1}^{k_{M-1}} \Bigg(
                \prod_{i=1}^{J_{a}^{(M-1)}-1-L_1^{\rmi}-\cdots -L_{M-2}^{\rmi}} 
                    \left(x^{(M-1)}_{\sigma_{M-1}(a)}-\widetilde{y}^{(M-1)}_{i} \right)
                \prod_{i=J_{a}^{(M-1)}+1-L_1^{\rmi}-\cdots -L_{M-2}^{\rmi}}^{\widetilde{L}_{M-1}}
                    \left(x^{(M-1)}_{\sigma_{M-1}(a)}-\widetilde{y}^{(M-1)}_{i}+1 \right)
            \Bigg)
            \\
            \times
            \prod_{a<b}^{k_{M-1}} 
                \frac{x^{(M-1)}_{\sigma_{M-1}(a)}-x^{(M-1)}_{\sigma_{M-1}(b)}-1}{x^{(M-1)}_{\sigma_{M-1}(a)}-x^{(M-1)}_{\sigma_{M-1}(b)}},
        \end{multline*}
        which is the form of the former ``final factor''. 
    \end{proof}

\subsection{Trigonometric version}
For the trigonometric case, we replace the $x$- and $y$-variables
by $u$- and $v$-variables and denote the partition function as
  \begin{multline}
            \psi\Big(\bm u^{(1)}, \dots,\bm u^{(M)} \Big| \left\{ \bm v^{(1)}_{\rmi}, \bm v^{(1)}_{\rmii} \right\} , \dots, \left\{ \bm v^{(M-1)}_{\rmi}, \bm v^{(M-1)}_{\rmii} \right\} , \bm v^{(M)} \Big| \\
            k_1, \dots, k_M \Big| \{L_1^{\rmi}, L_1^{\rmii} \}, \dots, \{L_{M-1}^{\rmi}, L_{M-1}^{\rmii} \}, L_M \Big| \bm I \Big).
\label{trigonometricpartitionfunctions}
        \end{multline}

The version of quantum (stochastic) $R$-matrix that we will use is defined by
\begin{align}
R(u,v)&=\sum_{i=1}^{M+1} (u-qv) E_{ii} \otimes E_{ii}+\sum_{i > j} q(u-v) E_{ii} \otimes E_{jj}
+\sum_{i < j} (u-v) E_{ii} \otimes E_{jj} \nonumber \\
&+\sum_{i > j} (1-q)u E_{ij} \otimes E_{ji}+\sum_{i < j} (1-q)v E_{ij} \otimes E_{ji}.
\end{align}
See Figure \ref{trigonometricrmatrix-components} for a graphical description
of the matrix elements.

The $R$-matrix satisfies the Yang-Baxter equation
\begin{align}
R_{23}(u_2,u_3)R_{13}(u_1,u_3)R_{12}(u_1,u_2)=
R_{12}(u_1,u_2)R_{13}(u_1,u_3)R_{23}(u_2,u_3),
\end{align}
acting on $V_1 \otimes V_2 \otimes V_3$.

\begin{figure}[ht]
    \centering
    \begin{tabular}{cc}
        \includegraphics{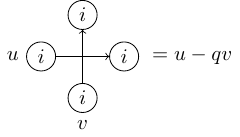}
        &
        \includegraphics{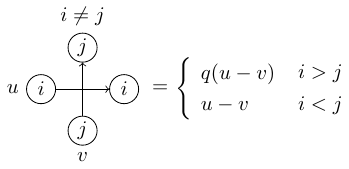}
    \end{tabular}
    \includegraphics{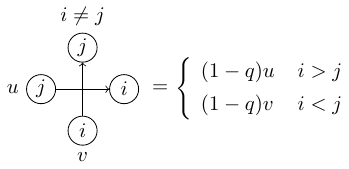}
 \caption{Matrix elements of the trigonometric $R$-matrix $R(u,v)$.}
\label{trigonometricrmatrix-components}
\end{figure}

The strategy to determine the explicit forms of the partition function
is identical the rational case.
Below, we state Korepin's lemma and give a characterization
of the partition function,
and then state that an extended version of the trigonometric weight function
satisfies the properties.

\begin{proposition}  \label{trigonometricKorepinlemma}
    Let $L^\rmi := L^{\rmi}_1+ \dots + L^{\rmi}_{M-1}+L_M$. 
    The functions $\psi$ satisfy the following properties:
    \begin{enumerate}
        \item \emph{Case A}: If $i_{L^{\rmi}}^{(M)} \neq M+1$,
the degree of $v_{L_M}^{(M)}$ in $\psi$ is at most $k_M-1$.
        \item $\psi$ is symmetric with respect to $\bm u^{(M)}$.
        \item \emph{Case A}: If $i^{(M)}_{L^{\rmi}} \neq M+1$, 
        \begin{multline}
            \psi\Big(\bm u^{(1)}, \dots,\bm u^{(M)} \Big| \left\{ \bm v^{(1)}_{\rmi}, \bm v^{(1)}_{\rmii} \right\} , \dots, \left\{ \bm v^{(M-1)}_{\rmi}, \bm v^{(M-1)}_{\rmii} \right\} , \bm v^{(M)} \Big|
            \\
            k_1, \dots, k_M \Big| \{L_1^{\rmi}, L_1^{\rmii} \}, \dots, \{L_{M-1}^{\rmi}, L_{M-1}^{\rmii} \}, L_M \Big| \bm I \Big) \Bigg|_{v^{(M)}_{L_M}=q^{-1}u^{(M)}_{k_M}}
            \\
            = (1-q) u^{(M)}_{k_M}\prod_{j=1}^{L_M-1} \left(u^{(M)}_{k_M}-v^{(M)}_j\right)
            \prod_{j=1}^{k_M-1} \left(qu^{(M)}_j-u^{(M)}_{k_M}\right)
            \\
            \times 
            \psi\Big(\bm u^{(1)}, \dots,\bm u^{(M)}  \setminus \{ u^{(M)}_{k_M} \} \Big| \left\{ \bm v^{(1)}_{\rmi}, \bm v^{(1)}_{\rmii} \right\} , \dots, \left\{ \bm v^{(M-1)}_{\rmi}, \{u^{(M)}_{k_M} \} \cup \bm v^{(M-1)}_{\rmii} \right\} , \bm v^{(M)} \setminus \{ v^{(M)}_{L_M} \} \Big|
            \\
            k_1, \dots,k_{M-1}, k_M-1 \Big| \{L_1^{\rmi}, L_1^{\rmii} \}, \dots, \{L_{M-1}^{\rmi}, L_{M-1}^{\rmii}+1 \}, L_M-1 \Big| \bm J \Big),
        \end{multline}
where the $\bm J$ are related to the $\bm I$ by 
        \begin{gather*}
           {\bm J}^{(i)}_{k_i} ={\bm I}^{(i)}_{k_i} \qquad 1 \leq i \leq M-1; \\
           {\bm J}^{(M)}_{k_M-1} ={\bm I}^{(M)}_{k_M} \setminus \{L^{\rmi} \}; \\
            \widetilde{{\bm J}}^{(i)}_{k_i} = \widetilde{{\bm I}}^{(i)}_{k_i} \qquad 1 \leq i \leq M-1.
        \end{gather*}
        \item \emph{Case B}: If $i^{(M)}_{L^{\rmi}} = M+1$,
        \begin{multline}
            \psi\Big(\bm u^{(1)}, \dots,\bm u^{(M)} \Big| \left\{ \bm v^{(1)}_{\rmi}, \bm v^{(1)}_{\rmii} \right\} , \dots, \left\{ \bm v^{(M-1)}_{\rmi}, \bm v^{(M-1)}_{\rmii} \right\} , \bm v^{(M)} \Big|
            \\
            k_1, \dots, k_M \Big| \{L_1^{\rmi}, L_1^{\rmii} \}, \dots, \{L_{M-1}^{\rmi}, L_{M-1}^{\rmii} \}, L_M \Big| \bm I \Big)
            \\
            =
            \prod_{j=1}^{k_M}\left(u_j - qv^{(M)}_{L_M}\right)
            \times 
            \psi\Big(\bm u^{(1)}, \dots,\bm u^{(M)} \Big| \left\{ \bm v^{(1)}_{\rmi}, \bm v^{(1)}_{\rmii} \right\} , \dots, \left\{ \bm v^{(M-1)}_{\rmi}, \bm v^{(M-1)}_{\rmii} \right\} , \bm v^{(M)} \setminus \{v^{(M)}_{L_M}\} \Big|
            \\
            k_1, \dots, k_M \Big| \{L_1^{\rmi}, L_1^{\rmii} \}, \dots, \{L_{M-1}^{\rmi}, L_{M-1}^{\rmii} \}, L_M-1 \Big| \bm J \Big),
        \end{multline}
where the $\bm J$ are related to the $\bm I$ by 
        \begin{gather*}
            {\bm J}^{(i)}_{k_i} = \{j | j \in {\bm I}^{(i)}_{k_i}, j\leq L^{\rmi} \} \cup \{j-1 | j \in {\bm I}^{(i)}_{k_i}, j > L^{\rmi} \} \qquad 1 \leq i \leq M-1; \\
            {\bm J}^{(M)}_{k_M} ={\bm I}^{(M)}_{k_M} 
            \\
            \widetilde{{\bm J}}^{(i)}_{k_i} = \widetilde{{\bm I}}^{(i)}_{k_i} \qquad 1 \leq i \leq M-1.
        \end{gather*}
        \item \emph{Initial condition}: If $k_M=1$ and $i^{(M)}_{L^{\rmi}} \neq M+1$,
            \begin{multline}
                \psi\Big(\bm u^{(1)}, \dots,\bm u^{(M)} \Big| \left\{ \bm v^{(1)}_{\rmi}, \bm v^{(1)}_{\rmii} \right\} , \dots, \left\{ \bm v^{(M-1)}_{\rmi}, \bm v^{(M-1)}_{\rmii} \right\} , \bm v^{(M)} \Big|
                \\
                k_1, \dots, k_M \Big| \{L_1^{\rmi}, L_1^{\rmii} \}, \dots, \{L_{M-1}^{\rmi}, L_{M-1}^{\rmii} \}, L_M \Big|\bm I \Big)
                \\
                =
                (1-q)u^{(M)}_1 \prod_{j=1}^{L_M-1}\left(u^{(M)}_1-v^{(M)}_j\right)
                \\
                \times 
                \psi\Big(\bm u^{(1)}, \dots,\bm u^{(M)} \Big| \left\{ \bm v^{(1)}_{\rmi}, \bm v^{(1)}_{\rmii} \right\} , \dots, \left\{ \bm v^{(M-2)}_{\rmi}, \bm v^{(M-2)}_{\rmii} \right\} , \bm v^{(M-1)}_{\rmi} \cup \{u^{(M)}_1\} \cup \bm v^{(M-1)}_{\rmii} \Big|
                \\
                k_1, \dots, k_M \Big| \{L_1^{\rmi}, L_1^{\rmii} \}, \dots, \{L_{M-2}^{\rmi}, L_{M-2}^{\rmii} \}, L^{\rmi}_{M-1} + L^{\rmii}_{M-1} + 1 \Big| \bm J \Big),
            \end{multline}
where the $\bm J$ are related to the $\bm I$ by
    \begin{gather*}
        {{\bm J}}_{k_i}^{(i)} =\{ j| j \in {{\bm I}}_{k_i}^{(i)}, j<L^{\rmi}   \} \cup \{ j-L_M+1| j \in {{\bm I}}_{k_i}^{(i)}, j \geq L^{\rmi} \}  \qquad 1 \leq i \leq M-1;
\\
        \widetilde{{\bm J}}_{k_i}^{(i)} = \widetilde{{\bm I}}_{k_i}^{(i)} \qquad 1 \leq i \leq M-1.
    \end{gather*}
    \end{enumerate}
\end{proposition}

\begin{definition}
We introduce the extended trigonometric weight functions as
\begin{multline} \label{6-extended-wavefunction}
    W\Big(\bm u^{(1)}, \dots,\bm u^{(M)} \Big| \left\{ \bm v^{(1)}_{\rmi}, \bm v^{(1)}_{\rmii} \right\} , \dots, \left\{ \bm v^{(M-1)}_{\rmi}, \bm v^{(M-1)}_{\rmii} \right\} , \bm v^{(M)} \Big|
    \\
    k_1, \dots, k_M \Big| \{L_1^{\rmi}, L_1^{\rmii} \}, \dots, \{L_{M-1}^{\rmi}, L_{M-1}^{\rmii} \}, L_M \Big| {\bm I} \Big)
    \\
    = \sum_{\sigma_1 \in S_{k_1}} \cdots \sum_{\sigma_M \in S_{k_M}}
    \prod_{p=1}^{M-1} \Bigg\{
        \prod_{a=1}^{k_p} \Bigg(
            \prod_{i=1}^{\widetilde{I}_{a}^{(p)}-1-L_1^{\rmi}-\cdots -L_{p-1}^{\rmi}} 
                \left(u^{(p)}_{\sigma_p(a)}-m_i^{L_p^{\rmi},k_{p+1},L_p^{\rmii}} \left(\bm v_{\rmi}^{(p)}\middle| \bm u^{(p+1)}_{\sigma_{p+1}} \middle| \bm v^{(p)}_{\rmii}\right) \right)(1-q)u^{(p)}_{\sigma_p(a)}
            \\
            \times 
            \prod_{i=\widetilde{I}_{a}^{(p)}+1-L_1^{\rmi}-\cdots -L_{p-1}^{\rmi}}^{L^{\rmi}_p+L^{\rmii}_p + k_{p+1}}
                \left(u^{(p)}_{\sigma_p(a)}-q m_i^{L_p^{\rmi},k_{p+1},L_p^{\rmii}} \left(\bm v_{\rmi}^{(p)}\middle| \bm u^{(p+1)}_{\sigma_{p+1}} \middle| \bm v^{(p)}_{\rmii}\right) \right)
        \Bigg)
        \prod_{a<b}^{k_p} 
            \frac{qu^{(p)}_{\sigma_p(a)}-u^{(p)}_{\sigma_p(b)}}{u^{(p)}_{\sigma_p(a)}-u^{(p)}_{\sigma_p(b)}}
    \Bigg\}
    \\
    \times 
    \prod_{a=1}^{k_M} \Bigg(
        \prod_{i=1}^{I_a^{(M)}-1-L_1^{\rmi}-\cdots - L_{M-1}^{\rmi}} 
            \left(u_{\sigma_M(a)}^{(M)}-v_i^{(M)}\right)
        \times (1-q) u^{(M)}_{\sigma_M(a)}
        \prod_{i=I_a^{(M)}+1-L_1^{\rmi}-\cdots - L_{M-1}^{\rmi}}^{L_M} 
            \left(u_{\sigma_M(a)}^{(M)}-q v_i^{(M)} \right) \Bigg)
    \\ \times 
    \prod_{a<b}^{k_M} 
        \frac{qu^{(M)}_{\sigma_M(a)}-u^{(M)}_{\sigma_M(b)}}{u^{(M)}_{\sigma_M(a)}-u^{(M)}_{\sigma_M(b)}},
\end{multline}
 where
    \begin{equation*}
        m_i^{L_p^{\rmi},k_{p+1},L_p^{\rmii}} \left(\bm v_{\rmi}^{(p)}\middle| \bm u^{(p+1)} \middle| \bm v^{(p)}_{\rmii}\right) = 
        \begin{cases}
            v^{(p)}_{\rmi,i} & 1 \leq i \leq L_p^{\rmi} \\
            u^{(p+1)}_{i-L_p^{\rmi}} & L_p^{\rmi}+1 \leq i \leq L_p^{\rmi}+k_{p+1} \\
            v^{(p)}_{\rmii,i-L_p^{\rmi}-k_{p+1}} & L_p^{\rmi}+k_{p+1}+1 \leq i \leq L_p^{\rmi}+k_{p+1}+L_p^{\rmii}
        \end{cases}.
    \end{equation*}
\end{definition}

\begin{theorem}
We have
\begin{multline}
 \psi \bigg( \bm u^{(1)} ,\dots, \bm u^{(M)} \bigg| 
        \{\bm v_{\rmi}^{(1)},\bm v_{\rmii}^{(1)} \} ,\dots, \{\bm v_{\rmi}^{(M-1)},\bm v_{\rmii}^{(M-1)} \}, \bm v^{(M)} \bigg| 
        k_1, \dots, k_M \bigg| \{L_1^{\rmi} , L_1^{\rmii} \}, \dots, \{L_{M-1}^{\rmi} , L_{M-1}^{\rmii} \}, L_M \bigg|  \bm I \bigg)
        \\
        =
 W \bigg( \bm u^{(1)} ,\dots, \bm u^{(M)} \bigg| 
        \{\bm v_{\rmi}^{(1)},\bm v_{\rmii}^{(1)} \} ,\dots, \{\bm v_{\rmi}^{(M-1)},\bm v_{\rmii}^{(M-1)} \}, \bm v^{(M)} \bigg| 
        k_1, \dots, k_M \bigg| \{L_1^{\rmi} , L_1^{\rmii} \}, \dots, \{L_{M-1}^{\rmi} , L_{M-1}^{\rmii} \}, L_M \bigg|  \bm I \bigg).
\end{multline}
\end{theorem}

\section{Elliptic partition functions}
In this section, we discuss the elliptic case.
The strategy to determine it is largely the same as
the rational/trigonometric case.
However, there are some additional parameters, and we need some more explanation,
so we discuss this case in this separate section.
We first recall the elliptic theta functions and the elliptic $R$-matrix associated with the ellipitic quantum group
\cite{Felder,FV1,Ca,FeldSch,ABF,eightvertex,FIJKMY,Fr,Konno,
JKOS,DJKMO,JKMO}. 

We fix two complex nonzero numbers $\tau,\gamma$.
The odd  theta function
\begin{align}
[z]=-\sum_{j \in \mathbb{Z}+1/2} e^{i \pi j^2 \tau+2 \pi i j (z+1/2)},
\end{align}
satisfies $[-z]=-[z]$
and the quasi-periodicities
    \begin{equation} \label{theta-quasiperiodicity}
        [z+1] = -[z]; \qquad [z+\tau] = e^{-2\pi i z - \pi i \tau}[z].
    \end{equation}

Let $\chi$ be a character, which is a group homomorphism from
multiplicative group $\Gamma=\mathbb{Z}+\tau \mathbb{Z}$ to $\mathbb{C}^\times$.
For each character $\chi$ and a positive integer $n$, we define an $n$-dimensional space
$\Theta_n(\chi)$ which consists of holomorphic functions $\phi(y)$ on $\mathbb{C}$
such that
\[
\phi(y+1)=\chi(1) \phi(y); \qquad \phi(y+\tau)=\chi(\tau)e^{-2 \pi i n y-\pi i n \tau} \phi(y).
\] 
We call the elements of $\Theta_n(\chi)$ \emph{elliptic polynomials}.
The following is an elliptic version of Lagrange interpolation.

\begin{proposition} \cite{FeldSch}
Let $P(y)$ and $Q(y)$ be elliptic polynomials in $\Theta_n(\chi)$
such that $\chi(1)=(-1)^n$ and $\chi(\tau)=(-1)^n e^\alpha$.
If $P(y_j)=Q(y_j)$ for $j=1,\dots,n$ such that $y_j-y_k \not\in \Gamma$ $(j \neq k)$,
$\sum_{k} y_k-\alpha \not\in \Gamma$,
then $P(y) \equiv Q(y)$.
\end{proposition}

Next, we introduce the dynamical $R$-matrix
of $E_{\tau,\gamma}(\mathfrak{gl}_{M+1})$ or $U_{q,p}(\mathfrak{gl}_{M+1})$
\cite{Felder,FV1,Ca,Fr,Konno,JKOS},
which is a function
$R(z,\lambda):\mathbb{C} \otimes \mathfrak{h}^* \to \mathrm{End}(V \otimes V)$
where $\mathfrak{h}$ and $\mathfrak{h}^*$ denote a Cartan subalgebra of $\mathfrak{gl}_{M+1}$ and its dual respectively.
The $(M+1)$-dimensional space $V$ spanned by basis $e_k$ ($k=1,\dots,M+1$)
corresponds to an $\mathfrak{h}$-module.
For $h=\mathrm{diag}(h^1,\dots,h^{M+1}) \in \mathfrak{h}$, we define $\mu_i \in \mathfrak{h}^*$ as $\mu_{i}(h)=h^i$. We have $V=\oplus_{i=1}^N V^{\mu_i}$ where $V^{\mu_i}=\mathbb{C} e_i$.

The dynamical $R$-matrix we use is
\begin{align}
R(z,w,\lambda)=&\sum_{i=1}^{M+1} [z-w-\gamma] E_{ii} \otimes E_{ii}
+\sum_{i \neq j} \alpha(z-w,\lambda_i-\lambda_j) E_{ii} \otimes E_{jj} \nonumber \\
&+\sum_{i \neq j} \beta(z-w,\lambda_i-\lambda_j) E_{ij} \otimes E_{ji},
\end{align}
where $\lambda_i=\lambda(E_{ii})$ ($i=1,\dots,M+1$) for $\lambda \in \mathfrak{h}^*$
are called the \emph{dynamical variables}, which were not present
in the rational/trigonometric cases discussed previously, 
and
\begin{align}
\alpha(z,\lambda):=\frac{[z][\lambda+\gamma]}{[\lambda]};
\qquad \beta(z,\lambda):=-\frac{[z+\lambda][\gamma]}{[\lambda]}.
\end{align}
See Figure \ref{dynamicalRmatrixfigure} for a graphical description
of the matrix elements of the dynamical $R$-matrix.
This description comes from the statistical physics origin of the dynamical $R$-matrix
as local Boltzmann weights of the face model
\cite{ABF,DJKMO,JKMO}.

\begin{figure*}[ht]
        \centering
        \begin{tabular}{cc}
            \includegraphics{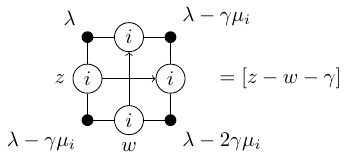}
            &
            \includegraphics{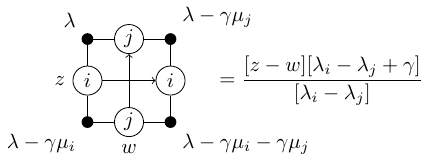}
        \end{tabular}

        \includegraphics{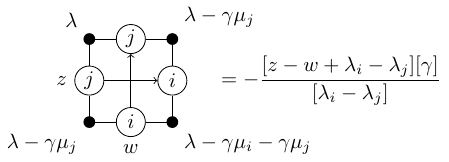}
\caption{The matrix elements of the dynamical $R$-matrix.
  Here $\lambda \in \mathfrak{h}^*$ is the dynamical parameter which can be viewed as $\lambda \in \C^n$ 
by the expansion $\lambda = \sum_{i=1}^n \lambda_i \mu_i$.
    The dynamical parameter associated to an $R$-matrix is the one that is in the upper left corner of the plaquet. 
    Then, the dynamical parameters of the remaining corners of the plaquet are deduced by subtracting $\gamma$ times the appropriate basis vector. 
}
\label{dynamicalRmatrixfigure}
    \end{figure*}

    The $R$ matrix satisfies the \emph{dynamical Yang-Baxter equation} acting on $V_1 \otimes V_2 \otimes V_3$
    \begin{multline} \label{dYBE}
        R_{23}(z_2-z_3,\lambda)\,R_{13}(z_1-z_3,\lambda-\gamma h_2)\,R_{12}(z_1-z_2,\lambda) \\ 
        = R_{12}(z_1-z_2,\lambda-\gamma h_3)\,R_{13}(z_1-z_3,\lambda-\gamma h_2)\,R_{23}(z_2-z_3,\lambda-\gamma h_1).
    \end{multline}
Here, the subscripts of $h$ refer to the space on which it acts.
For example,
\begin{align}
R_{12}(z_1-z_2,\lambda-\gamma h_3)(v_1 \otimes v_2 \otimes v_3)
=R(z_1-z_2,\lambda-\gamma \alpha)v_1 \otimes v_2 \otimes v_3,
\end{align}
if $v_3 \in V^\alpha$.
Historically, the dynamical Yang-Baxter equation corresponds to the star-triangle relation.
See Figure \ref{figuredynamicalyangbaxter} for a graphical description of \eqref{dYBE}.
    \begin{figure}[ht]
        \centering
        \includegraphics{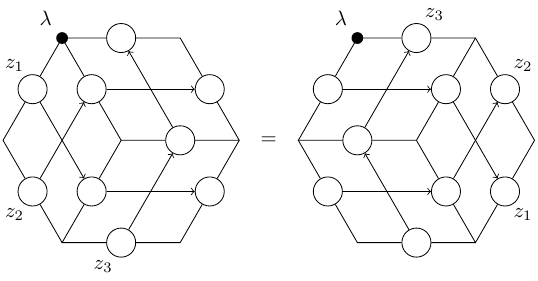}
        \caption{The dynamical Yang-Baxter equation \eqref{dYBE}
or equivalently the star-triangle relation.}
\label{figuredynamicalyangbaxter}
    \end{figure}

Next, we introduce elliptic partition functions.
We simplify the notation for elliptic weights of the dynamical $R$ matrix as Figure \ref{simplificationdynamicalRmatrix}.
\begin{figure}[ht]
\centering
    \includegraphics{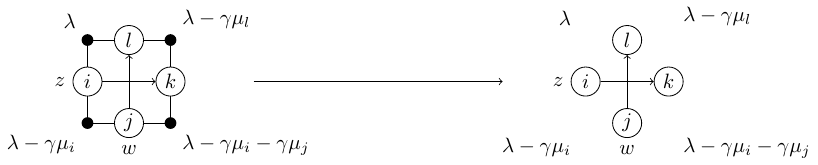}
\caption{Simplified drawing of the dynamical $R$-matrix.}
\label{simplificationdynamicalRmatrix}
\end{figure}
Then, the elliptic partition function is introduced pictorially in Figure \ref{figureellipticpartitionfunctions}, which we denote as
\begin{equation} \label{ellipticwavefunction}
        \psi\Big(\bm z^{(1)}, \dots,\bm z^{(M)} \Big| \left\{ \bm w^{(1)}_{\rmi}, \bm w^{(1)}_{\rmii} \right\} , \dots, \left\{ \bm w^{(M-1)}_{\rmi}, \bm w^{(M-1)}_{\rmii} \right\} , \bm w^{(M)} \Big|
        k_1, \dots, k_M \Big| \{L_1^{\rmi}, L_1^{\rmii} \}, \dots, \{L_{M-1}^{\rmi}, L_{M-1}^{\rmii} \}, L_M \Big| {\bm I} \Big| \lambda\Big).
\end{equation}
\begin{figure}[ht]
    \includegraphics[width=\textwidth]{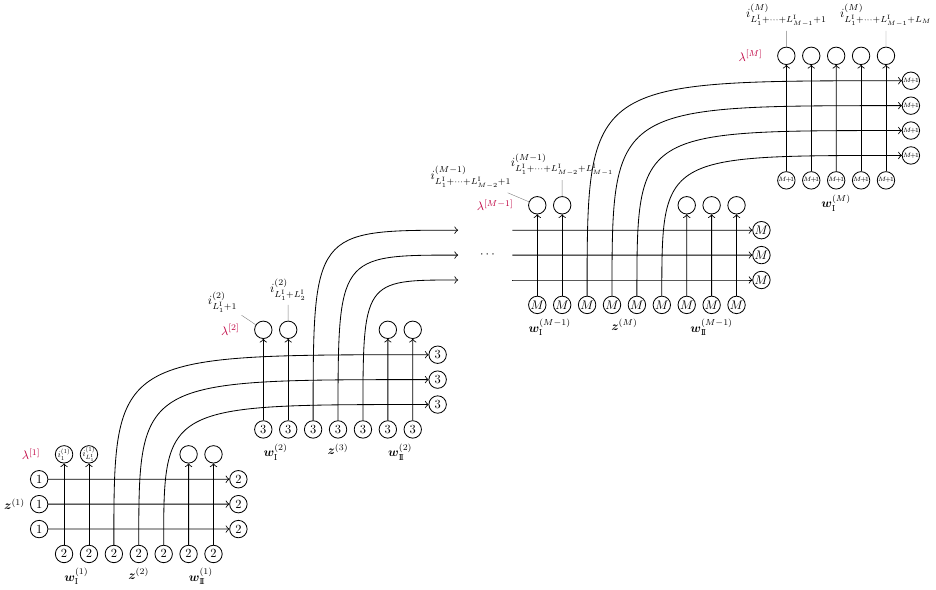}
\caption{The elliptic partition function.}
\label{figureellipticpartitionfunctions}
\end{figure}
Note that compared with the rational/trigonometric case, the dynamical $R$-matrix and hence
the partition functions also depend on $\lambda \in \mathfrak{h}^*$.

We introduce the notation
 \begin{align}
        \lambda^{[p]} := \lambda + \gamma \sum_{k=1}^{M-p} \sum_{j=1}^{L^{\rmi}_{M-k}} \mu_{i^{(M-k)}_{L^{\rmi}_1+\dots + L^{\rmi}_{M-k-1}+j}} \qquad  1 \leq p \leq M,
\label{shiftedlambda}
\end{align}
which will also be used for the elliptic multisymmetric functions later.
Note $\lambda^{[M]}=\lambda$.
In the partition functions, this is associated with the northwest corner in the top layer ($M$-th layer).
Note that a parameter associated with some corner changes if one crosses a line and move to another corner of the plaquet,
and the parameter associated with the northwest corner in the $j$-th layer (counted from bottom) becomes $\lambda^{[j]}$.

Finally, we introduce the notation $C^{(p)}(k,l)$ by
\begin{align}
       C^{(p)}(k,l):=|\{(a,j) | i^{(a)}_j=l,p \leq a \leq M , j \leq k \}|. \label{dynamicalsymbol}
\end{align}
The $C^{(p)}(k,l)$ notation has the meaning of counting the number of sites coloured by \col{l} until site number $k$ in the top $(M-p+1)$ layers.
See Figure \ref{figuresymbol}.

\begin{figure}[ht]
    \includegraphics{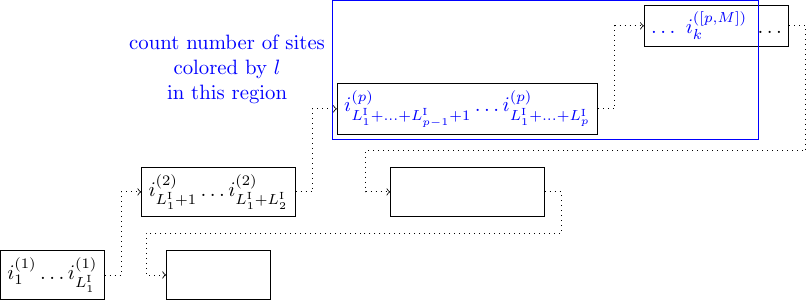}
\caption{The symbol $C^{(p)}(k,l)$ which
may be explained diagrammatically as follows. 
Among the sites on level $p$ or greater, but with site number no greater than $k$ (highlighted in blue), count the number of sites coloured by $l$. 
}
\label{figuresymbol}
\end{figure}

 \begin{proposition}
        Let $L^{\rmi} := L^{\rmi}_1 + \cdots + L^{\rmi}_{M-1} + L_M$.
        The functions $\psi$ satisfy the following properties:
        \begin{enumerate}
            \item \emph{Case A}: If $i_{L^{\rmi}}^{(M)} \neq M+1$,
$\psi$ are elliptic polynomials of $w^{(M)}_{L_M}$ in $\Theta_{k_M}(\chi)$ with the following quasiperiodicities
            \begin{align}
                \psi(w^{(M)}_{L_M}+1) &= (-1)^{k_M} \psi(w^{(M)}_{L_M}),
            \end{align}
            \begin{multline}
                \psi(w^{(M)}_{L_M}+\tau) = (-1)^{k_M} 
                \exp\Bigg(-2\pi i\bigg(k_M w^{(M)}_{L_M} - \sum_{i=1}^{k_M} z^{(M)}_i + \lambda_{i^{(M)}_{L^{\rmi} }} - \lambda_{M+1}
                \\
                \qquad  + \gamma\left(L_M-C^{(M)}(L^{\rmi},i^{(M)}_{L^{\rmi} } )\right)\bigg) - \pi i k_M \tau \Bigg)
                \psi(w^{(M)}_{L_M}).
            \end{multline}
            \item $\psi$ is symmetric with respect to $\bm z^{(M)}_{k_M}$.
            \item \emph{Case A}: If $i^{(M)}_{L^{\rmi}} \neq M+1$, 
            \begin{multline}
                \psi\Big(\bm z^{(1)}, \dots,\bm z^{(M)} \Big| \left\{ \bm w^{(1)}_{\rmi}, \bm w^{(1)}_{\rmii} \right\} , \dots, \left\{ \bm w^{(M-1)}_{\rmi}, \bm w^{(M-1)}_{\rmii} \right\} , \bm w^{(M)} \Big|
                \\
                k_1, \dots, k_M \Big| \{L_1^{\rmi}, L_1^{\rmii} \}, \dots, \{L_{M-1}^{\rmi}, L_{M-1}^{\rmii} \}, L_M \Big| {\bm I} \Big| \lambda\Big) \Bigg|_{w^{(M)}_{L_M}=z^{(M)}_{k_M}-\gamma}
                \\
                = \prod_{j=1}^{L_M-1} \left[z^{(M)}_{k_M}-w^{(M)}_j\right]
                \prod_{j=1}^{k_M-1} \left[z^{(M)}_j-z^{(M)}_{k_M}+\gamma\right]
                \frac{
                    [\gamma]\left[
                        \lambda_{M+1} - \lambda_{i^{(M)}_{L^{\rmi}}}+\gamma \left(
                        k_M-L_M+C^{(M)} \left(
                        L^{\rmi},i^{(M)}_{L^{\rmi}} \right)
                        \right)
                    \right]
                }{
                    \left[
                        \lambda_{i^{(M)}_{L^{\rmi}}} - \lambda_{M+1} +\gamma \left(1-C^{(M)} \left(
                        L^{\rmi},i^{(M)}_{L^{\rmi}} \right)
                        \right)
                    \right]}
                \\
                \times 
                \psi\Big(\bm z^{(1)}, \dots,\bm z^{(M)}  \setminus \{ z^{(M)}_{k_M} \} \Big| \left\{ \bm w^{(1)}_{\rmi}, \bm w^{(1)}_{\rmii} \right\} , \dots, \left\{ \bm w^{(M-1)}_{\rmi}, \{z^{(M)}_{k_M} \} \cup \bm w^{(M-1)}_{\rmii} \right\} , \bm w^{(M)} \setminus \{ w^{(M)}_{L_M} \} \Big|
                \\
                k_1, \dots,k_{M-1}, k_M-1 \Big| \{L_1^{\rmi}, L_1^{\rmii} \}, \dots, \{L_{M-1}^{\rmi}, L_{M-1}^{\rmii} \}, L_M-1 \Big| \bm J \Big| \lambda\Big),
            \end{multline}
where the $\bm J$ are related to the $\bm I$ by 
        \begin{gather*}
           {\bm J}^{(i)}_{k_i} ={\bm I}^{(i)}_{k_i} \qquad 1 \leq i \leq M-1; \\
           {\bm J}^{(M)}_{k_M-1} ={\bm I}^{(M)}_{k_M} \setminus \{L^{\rmi}\}; \\
            \widetilde{{\bm J}}^{(i)}_{k_i} = \widetilde{{\bm I}}^{(i)}_{k_i} \qquad 1 \leq i \leq M-1.
        \end{gather*}
            \item \emph{Case B}: If $i^{(M)}_{L^{\rmi}} = M+1$,
            \begin{multline}
                \psi\Big(\bm z^{(1)}, \dots,\bm z^{(M)} \Big| \left\{ \bm w^{(1)}_{\rmi}, \bm w^{(1)}_{\rmii} \right\} , \dots, \left\{ \bm w^{(M-1)}_{\rmi}, \bm w^{(M-1)}_{\rmii} \right\} , \bm w^{(M)} \Big|
                \\
                k_1, \dots, k_M \Big| \{L_1^{\rmi}, L_1^{\rmii} \}, \dots, \{L_{M-1}^{\rmi}, L_{M-1}^{\rmii} \}, L_M \Big|\bm I \Big| \lambda\Big)
                \\
                = 
                \prod_{j=1}^{k_M}\left[z_j - w^{(M)}_{L_M} - \gamma\right]
                \\
                \times
                \psi\Big(\bm z^{(1)}, \dots,\bm z^{(M)} \Big| \left\{ \bm w^{(1)}_{\rmi}, \bm w^{(1)}_{\rmii} \right\} , \dots, \left\{ \bm w^{(M-1)}_{\rmi}, \bm w^{(M-1)}_{\rmii} \right\} , \bm w^{(M)} \setminus \{w^{(M)}_{L_M}\} \Big|
                \\
                k_1, \dots, k_M \Big| \{L_1^{\rmi}, L_1^{\rmii} \}, \dots, \{L_{M-1}^{\rmi}, L_{M-1}^{\rmii} \}, L_M-1 \Big| \bm J \Big| \lambda\Big),
            \end{multline}
 where the $\bm J$ are related to the $\bm I$ by
        \begin{gather*}
            {\bm J}^{(i)}_{k_i} = \{j | j \in {\bm I}^{(i)}_{k_i}, j\leq L^{\rmi} \} \cup \{j-1 | j \in {\bm I}^{(i)}_{k_i}, j > L^{\rmi} \} \qquad 1 \leq i \leq M-1; \\
           {\bm J}^{(M)}_{k_M} ={\bm I}^{(M)}_{k_M}
           \\
            \widetilde{{\bm J}}^{(i)}_{k_i} = \widetilde{{\bm I}}^{(i)}_{k_i} \qquad 1 \leq i \leq M-1.
        \end{gather*}
            \item \emph{Initial condition}: If $k_M=1$ and $i^{(M)}_{L^{\rmi}} \neq M+1$,
            \begin{multline}
                \psi\Big(\bm z^{(1)}, \dots,\bm z^{(M)} \Big| \left\{ \bm w^{(1)}_{\rmi}, \bm w^{(1)}_{\rmii} \right\} , \dots, \left\{ \bm w^{(M-1)}_{\rmi}, \bm w^{(M-1)}_{\rmii} \right\} , \bm w^{(M)} \Big|
                \\
                k_1, \dots, k_M \Big| \{L_1^{\rmi}, L_1^{\rmii} \}, \dots, \{L_{M-1}^{\rmi}, L_{M-1}^{\rmii} \}, L_M \Big|\bm I \Big| \lambda\Big)
                \\
                =
                \frac{
                    [\gamma]\left[
                        z^{(M)}_1-w^{(M)}_{L_M}+\lambda_{M+1}-\lambda_{i^{(M)}_{L^{\rmi}}}-\gamma(L_M-1)
                        \right]
                }{
                    \left[
                        \lambda_{i^{(M)}_{L^{\rmi}}}-\lambda_{M+1}
                    \right]
                } \prod_{j=1}^{L_M-1}\left[z^{(M)}_1-w^{(M)}_j\right]
                \\
                \times 
                \psi\Big(\bm z^{(1)}, \dots,\bm z^{(M)} \Big| \left\{ \bm w^{(1)}_{\rmi}, \bm w^{(1)}_{\rmii} \right\} , \dots, \left\{ \bm w^{(M-2)}_{\rmi}, \bm w^{(M-2)}_{\rmii} \right\} , \bm w^{(M-1)}_{\rmi} \cup \{z^{(M)}_1\} \cup \bm w^{(M-1)}_{\rmii} \Big|
                \\
                k_1, \dots, k_M \Big| \{L_1^{\rmi}, L_1^{\rmii} \}, \dots, \{L_{M-2}^{\rmi}, L_{M-2}^{\rmii} \}, L^{\rmi}_{M-1} + L^{\rmii}_{M-1} + 1 \Big| \bm J \Big| \lambda^{[M-1]}\Big),
            \end{multline}
where the $\bm J$ are related to the $\bm I$ by
    \begin{gather*}
        {{\bm J}}_{k_i}^{(i)} =\{ j| j \in {{\bm I}}_{k_i}^{(i)}, j<L^{\rmi}   \} \cup \{ j-L_M+1| j \in {{\bm I}}_{k_i}^{(i)}, j \geq L^{\rmi} \}  \qquad 1 \leq i \leq M-1;
\\
        \widetilde{{\bm J}}_{k_i}^{(i)} = \widetilde{{\bm I}}_{k_i}^{(i)} \qquad 1 \leq i \leq M-1.
    \end{gather*}
        \end{enumerate}
    \end{proposition}

\begin{proof}
   \emph{Property 1:} The quasiperiodicity of the partition function is proved as follows. 
    The dependence of $\psi$ on $w_{L_M}^{(M)}$ is contained in the rightmost column of the partition function, and the color assigned to this quantum space is $i := i^{(M)}_{L^\rmi} \neq M+1$. 
    By the ice rule, there must be exactly one auxiliary space which is assigned color $i$; let us denote by $l$ the label assigned to this space. 
    Therefore, the rightmost column assumes the form in Figure~\ref{ell-wavefunction-col}, where $\Lambda := \lambda - \gamma \sum_{j=1}^{M+1} C^{(M)}( L^{\rmi}, j) \mu_j + \gamma \mu_i = \sum_{j=1}^{M+1} \Lambda_j \mu_j$, indicating the dynamical variable associated to the upper left of this column. 
    From these factors, we collect the factors which depend on $w^{(M)}_{L_M}$ and denote their product as $h_l(w^{(M)}_{L_M})$, so
    \begin{multline}
        h_l(w^{(M)}_{L_M}) := 
        \prod_{j=1}^{l-1}
        [z_j^{(M)}-w^{(M)}_{L_M}]
        \prod_{j=l+1}^{k_M}
        [z_j^{(M)}-w^{(M)}_{L_M}-\gamma]
        \times \\ \times
        \left[z_j^{(M)}-w^{(M)}_{L_M}+\lambda_{M+1}-\lambda_{i^{(M)}_{L^\rmi}} + \gamma\left(C^{(M)}(L^{\rmi},i^{(M)}_{L^\rmi})-C^{(M)}(L^{\rmi},M+1)-l\right)\right].
    \end{multline}
    From the periodicity of the odd theta function \eqref{theta-quasiperiodicity}, we see immediately that 
    \[
        h_l(w^{(M)}_{L_M} + 1) = (-1)^{k_M}h_l(w^{(M)}_{L_M}).
    \]
    Further, the second periodicity of the odd theta function gives
    \begin{multline}
        h_l(w^{(M)}_{L_M} + \tau) = (-1)^{k_M} \exp\bigg( 
            -2\pi i \bigg(k_M w^{(M)}_{L_M} - \sum_{i=1}^{k_M}z_i +\lambda_{i^{(M)}_{L^\rmi}} - \lambda_{M+1} + \\ \gamma\left(k_M + C^{(M)}(L^\rmi,M+1)-C^{(M)}(L^\rmi,i^{(M)}_{L^\rmi}) \right)\bigg) - \pi i k_M\tau 
        \bigg)h_l(w^{(M)}_{L_M}).
    \end{multline}
    Using $C^{(M)}(L^\rmi,M+1) = L_M - k_M$, we obtain the quasiperiodicity of $h_l(w^{(M)}_{L_M})$. 
    Then, since all dependence of $\psi$ on $w^{(M)}_{L_M}$ is contained in $h_l(w^{(M)}_{L_M})$, this gives the quasiperiodicity of $\psi$.

    \begin{figure}
        \centering
        \includegraphics{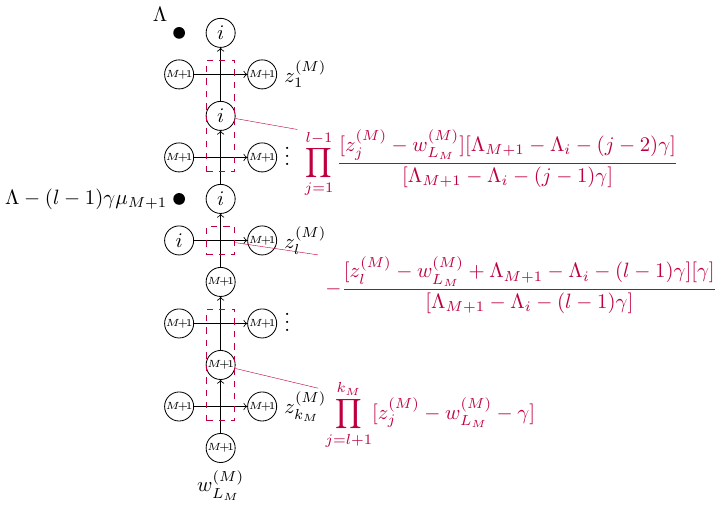}
        \caption{The evaluation of the rightmost column of the partition function $\psi$, in Case A.
        Here $\Lambda := \lambda - \gamma \sum_{j=1}^{M+1} C^{(M)}( L^{\rmi}, j) \mu_j + \gamma \mu_i$ and $i=i^{(M)}_{L^\rmi}$. }
        \label{ell-wavefunction-col}
    \end{figure}

    \emph{Property 2:} the proof runs parallel with the rational/trigonometric case.

    \emph{Property 3:} 
    \begin{figure}
        \label{ell-glM_caseA}
        \includegraphics[width=\textwidth]{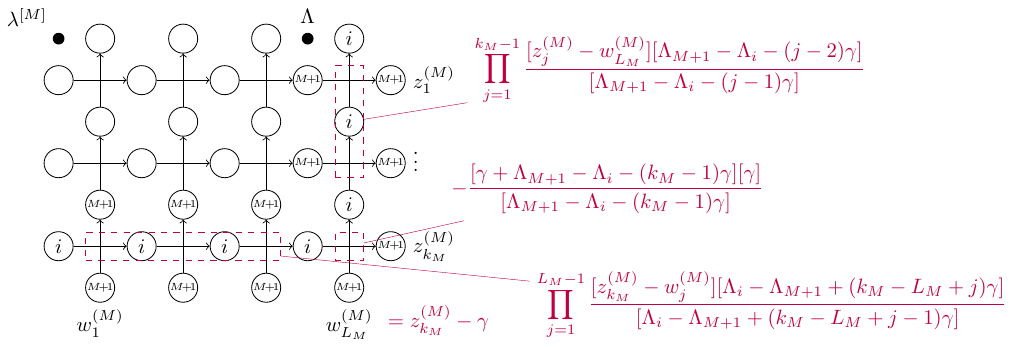}
        \caption{Diagram describing the evaluation of the elliptic partition function for Case A.}
    \end{figure}
    Recall the definitions of $\Lambda$ and $i$ from the proof of Property 1 above. 

    From Figure~\ref{ell-glM_caseA}, we obtain three factors. 
    The first is given by 
    \[
        \prod_{j=1}^{k_M-1} \frac{[z^{(M)}_j-z^{(M)}_{k_M}+\gamma][\Lambda_{M+1} - \Lambda_i - (j-2) \gamma]}{[\Lambda_{M+1} - \Lambda_i - (j-1) \gamma]}.
    \]
    The factors of this product cancel with one another to become
    \[
        \left(\prod_{j=1}^{k_M-1} [z^{(M)}_j-z^{(M)}_{k_M}+\gamma]\right) 
        \frac{[\Lambda_{M+1} - \Lambda_i + \gamma]}{[\Lambda_{M+1} - \Lambda_i - (k_M-2) \gamma]}.
    \]
    The second factor is 
    \[
        -\frac{[\Lambda_{M+1} - \Lambda_i - (k_M-2)\gamma][\gamma]}{[\Lambda_{M+1} - \Lambda_i - (k_M-1)\gamma]}.
    \]
    The third factor is 
    \[
        \prod_{j=1}^{L_M-1}\frac{[z_{k_M}^{(M)}-w_{j}^{(M)}][\Lambda_i-\Lambda_{M+1}+(k_M-L_M+j)\gamma]}{[\Lambda_i-\Lambda_{M+1}+(k_M-L_M+j-1)\gamma]}.
    \]
    After many cancellations, this becomes
    \[
        \left(\prod_{j=1}^{L_M-1}[z_{k_M}^{(M)}-w_{j}^{(M)}]\right)\frac{[\Lambda_{M+1}-\Lambda_i-(k_M-1)\gamma]}{[\Lambda_{M+1}-\Lambda_i-(k_M-L_M)\gamma]},
    \]
    where we have also used $[-z] = -[z]$.
    Combining the three factors gives 
    \begin{multline*}
        \left(\prod_{j=1}^{k_M-1} [z^{(M)}_j-z^{(M)}_{k_M}+\gamma]\right) 
        \frac{[\Lambda_{M+1} - \Lambda_i + \gamma]}{\cancel{[\Lambda_{M+1} - \Lambda_i - (k_M-2) \gamma]}} 
        \times 
        -\frac{\cancel{[\Lambda_{M+1} - \Lambda_i - (k_M-2)\gamma]}[\gamma]}{\cancel{[\Lambda_{M+1} - \Lambda_i - (k_M-1)\gamma]}}
        \\
        \times 
        \left(\prod_{j=1}^{L_M-1}[z_{k_M}^{(M)}-w_{j}^{(M)}]\right)\frac{\cancel{[\Lambda_{M+1}-\Lambda_i-(k_M-1)\gamma]}}{[\Lambda_{M+1}-\Lambda_i-(k_M-L_M)\gamma]},
    \end{multline*}
    which becomes 
    \begin{equation*}
        -\left(\prod_{j=1}^{k_M-1} [z^{(M)}_j-z^{(M)}_{k_M}+\gamma]\right) 
        \left(\prod_{j=1}^{L_M-1}[z_{k_M}^{(M)}-w_{j}^{(M)}]\right)
        \frac{[\Lambda_{M+1} - \Lambda_i + \gamma][\gamma]}{[\Lambda_{M+1}-\Lambda_i-(k_M-L_M)\gamma]}.
    \end{equation*}
    Finally, from the definition of $\Lambda$, we have 
    \[
        \Lambda_{M+1} = \lambda_{M+1} - \gamma C^{(M)}(L^\rmi,M+1) ;\qquad 
        \Lambda_{i} = \lambda_{i} - \gamma C^{(M)}(L^\rmi,i) +\gamma.
    \]
    This, with $C^{(M)}(L^\rmi,M+1)=L_M-k_M$, proves property 3.
    
    \emph{Property 4 and 5:} The proof runs parallel to the rational/trigonometric case. 
\end{proof}

Next, we introduce the elliptic multisymmetric functions.

\begin{definition}
We define the extended elliptic weight functions as
\begin{multline} \label{5-extended-wavefunction}
        \psi\Big(\bm z^{(1)}, \dots,\bm z^{(M)} \Big| \left\{ \bm w^{(1)}_{\rmi}, \bm w^{(1)}_{\rmii} \right\} , \dots, \left\{ \bm w^{(M-1)}_{\rmi}, \bm w^{(M-1)}_{\rmii} \right\} , \bm w^{(M)} \Big|
        \\
        k_1, \dots, k_M \Big| \{L_1^{\rmi}, L_1^{\rmii} \}, \dots, \{L_{M-1}^{\rmi}, L_{M-1}^{\rmii} \}, L_M \Big| \bm I \Big| \lambda\Big)
        \\
        = [\gamma]^{k_1 + \dots + k_M} \sum_{\sigma_1 \in S_{k_1}} \cdots \sum_{\sigma_M \in S_{k_M}}
        \prod_{p=1}^{M-1} \Bigg\{
            \prod_{a=1}^{k_p} \Bigg(
                \prod_{i=1}^{\widetilde{I}_{a}^{(p)}-1-L_1^{\rmi}-\cdots -L_{p-1}^{\rmi}} 
                    \left[z^{(p)}_{\sigma_p(a)}-m_i^{L_p^{\rmi},k_{p+1},L_p^{\rmii}} \left(\bm w_{\rmi}^{(p)}\middle| \bm z^{(p+1)}_{\sigma_{p+1}} \middle| \bm w^{(p)}_{\rmii}\right) \right]
                \\
                \times 
                \frac{\left[z^{(p)}_{\sigma_p(a)}-m^{L^{\rmi}_p,k_{p+1},L^{\rmii}_p}_{\widetilde{I}^{(p)}_a-L^{\rmi}_1-\dots-L^{\rmi}_{p-1}}(w^{(p)}_{\rmi}|z^{(p+1)}_{\sigma_{p+1}}|w^{(p)}_{\rmii})
                +\lambda^{[p]}_{p+1}-\lambda^{[p]}_{i^{([p,M])}_{I^{(p)}_a}}+\gamma\left( C^{(p)}(I^{(p)}_a,i^{([p,M])}_{I^{(p)}_a})-C^{(p)}(I^{(p)}_a,p+1)-1\right)\right]}{
                    [\lambda^{[p]}_{i^{([p,M])}_{I^{(p)}_a}}-\lambda^{[p]}_{p+1}+\gamma(1-C^{(p)}(I^{(p)}_a,i^{([p,M])}_{I^{(p)}_a}))]
                }
                \\
                \times 
                \prod_{i=\widetilde{I}_{a}^{(p)}+1-L_1^{\rmi}-\cdots -L_{p-1}^{\rmi}}^{L^{\rmi}_p+L^{\rmii}_p + k_{p+1}}
                    \left[z^{(p)}_{\sigma_p(a)}-m_i^{L_p^{\rmi},k_{p+1},L_p^{\rmii}} \left(\bm w_{\rmi}^{(p)}\middle| \bm z^{(p+1)}_{\sigma_{p+1}} \middle| \bm w^{(p)}_{\rmii}\right)-\gamma \right]
            \Bigg)
            \prod_{a<b}^{k_p} 
                \frac{[z^{(p)}_{\sigma_p(a)}-z^{(p)}_{\sigma_p(b)}+\gamma]}{[z^{(p)}_{\sigma_p(a)}-z^{(p)}_{\sigma_p(b)}]}
        \Bigg\}
        \\
        \times 
        \prod_{a=1}^{k_M} \Bigg(
            \prod_{i=1}^{I_a^{(M)}-1-L_1^{\rmi}-\cdots - L_{M-1}^{\rmi}} 
                \left[z_{\sigma_M(a)}^{(M)}-w_i^{(M)}\right]
            \prod_{i=I_a^{(M)}+1-L_1^{\rmi}-\cdots - L_{M-1}^{\rmi}}^{L_M} 
                \left[z_{\sigma_M(a)}^{(M)}-w_i^{(M)}-\gamma \right]
        \\
        \times 
        \frac{\left[z^{(M)}_{\sigma_M(a)}-w^{(M)}_{I^{(M)}_a-L^{\rmi}_1-\dots-L^{\rmi}_{M-1}}
        +\lambda^{[M]}_{M+1}-\lambda^{[M]}_{i^{(M)}_{I^{(M)}_a}}+\gamma\left( C^{(M)}(I^{(M)}_a,i^{(M)}_{I^{(M)}_a})-C^{(M)}(I^{(M)}_a,M+1)-1\right)\right]}{
            [\lambda^{[M]}_{i^{(M)}_{I^{(M)}_a}}-\lambda^{[M]}_{M+1}+\gamma(1-C^{(M)}(I^{(M)}_a,i^{(M)}_{I^{(M)}_a}))]
        }
        \Bigg)
        \\ \times 
        \prod_{a<b}^{k_M} 
            \frac{[z^{(M)}_{\sigma_M(a)}-z^{(M)}_{\sigma_M(b)}+\gamma]}{[z^{(M)}_{\sigma_M(a)}-z^{(M)}_{\sigma_M(b)}]},
    \end{multline}
    where
    \begin{equation*}
        m_i^{L_p^{\rmi},k_{p+1},L_p^{\rmii}} \left(\bm w_{\rmi}^{(p)}\middle| \bm z^{(p+1)} \middle| \bm w^{(p)}_{\rmii}\right) = 
        \begin{cases}
            w^{(p)}_{\rmi,i} & 1 \leq i \leq L_p^{\rmi} \\
            z^{(p+1)}_{i-L_p^{\rmi}} & L_p^{\rmi}+1 \leq i \leq L_p^{\rmi}+k_{p+1} \\
            w^{(p)}_{\rmii,i-L_p^{\rmi}-k_{p+1}} & L_p^{\rmi}+k_{p+1}+1 \leq i \leq L_p^{\rmi}+k_{p+1}+L_p^{\rmii}
        \end{cases},
    \end{equation*}
 and $C^{(p)}(k,l)$ is as defined in \eqref{dynamicalsymbol}. 
The symbol $i_k^{([p,M])}$ means $i_k^{(j)}$ for some $j \in \{p,p+1,\dots,M \}$.
\end{definition}

\begin{theorem}
We have
\begin{multline}
        \psi\Big(\bm z^{(1)}, \dots,\bm z^{(M)} \Big| \left\{ \bm w^{(1)}_{\rmi}, \bm w^{(1)}_{\rmii} \right\} , \dots, \left\{ \bm w^{(M-1)}_{\rmi}, \bm w^{(M-1)}_{\rmii} \right\} , \bm w^{(M)} \Big|
        \\
        k_1, \dots, k_M \Big| \{L_1^{\rmi}, L_1^{\rmii} \}, \dots, \{L_{M-1}^{\rmi}, L_{M-1}^{\rmii} \}, L_M \Big| \bm I \Big| \lambda\Big)
        \\
        =  W\Big(\bm z^{(1)}, \dots,\bm z^{(M)} \Big| \left\{ \bm w^{(1)}_{\rmi}, \bm w^{(1)}_{\rmii} \right\} , \dots, \left\{ \bm w^{(M-1)}_{\rmi}, \bm w^{(M-1)}_{\rmii} \right\} , \bm w^{(M)} \Big|
        \\
        k_1, \dots, k_M \Big| \{L_1^{\rmi}, L_1^{\rmii} \}, \dots, \{L_{M-1}^{\rmi}, L_{M-1}^{\rmii} \}, L_M \Big| \bm I \Big| \lambda\Big).
\end{multline}
\end{theorem}

 \begin{proof}
The proof is the same as the rational/trigonometric case.
Note that since there are dynamical parameters, one needs to take into account the relations between the symbols \eqref{dynamicalsymbol}
associated with $\bm I$ and $\bm J$.
        For Property 3, we use the following relations
        \begin{align*}
            &i^{([p,M])}_{I^{(p)}_a} =j^{([p,M])}_{J^{(p)}_a}  \qquad p=1, \dots, M-1
            \\
            &\qquad \Rightarrow \lambda^{[p]}_{i^{([p,M])}_{I^{(p)}_a}} = \lambda^{[p]}_{j^{([p,M])}_{J^{(p)}_a}} 
            \\
            &\qquad C^{(p)}(I^{(p)}_a, i^{([p,M])}_{I^{(p)}_a}) = C^{(p)}(J^{(p)}_a, j^{([p,M])}_{J^{(p)}_a}).
        \end{align*}
 For Property 4,
        we must take into account the following relations
        \begin{align*}
            \lambda^{[p]}_{i^{([p,M])}_{I^{(p)}_a}} &= \lambda^{[p]}_{j^{([p,M])}_{J^{(p)}_a}} 
            \\
            C^{(p)}(I^{(p)}_a, i^{([p,M])}_{I^{(p)}_a}) &= C^{(p)}(J^{(p)}_a, j^{([p,M])}_{J^{(p)}_a})
            \\
            C^{(p)}(I^{(p)}_a, p+1) &= C^{(p)}(J^{(p)}_a, p+1) \qquad 1 \leq p+1 \leq M.
        \end{align*}
 For Property 5,
        we use the following relations
        \begin{align*}
            i^{([p,M])}_{I^{(p)}_a} = j^{([p,M-1])}_{J^{(p)}_a}
            \\
            C^{(p)}(I^{(p)}_a, i^{([p,M])}_{I^{(p)}_a}) &= C^{(p)}(J^{(p)}_a, j^{([p,M-1])}_{J^{(p)}_a})
            \\
            C^{(p)}(I^{(p)}_a, p+1) &= C^{(p)}(J^{(p)}_a, p+1) \qquad 1 \leq p+1 \leq M.
        \end{align*}
    \end{proof}

\section{Special cases and elliptic weight functions}
Let us compare special cases of the elliptic multisymmetric functions introduced
in this paper with the elliptic weight functions.
The original elliptic weight functions correspond to the case $L_i^{\rmi}=L_i^{\rmii}=0$ for $i=1,\dots,M-1$, so
${\bm w}_{\rmi}^{(i)}={\bm w}_{\rmii}^{(i)}=\emptyset$.
In this case, only the $p=M$ case of $C^{(p)}(k,l)$ and $i_{j}^{([p,M])}$ appears, so we just denote these as $C(k,l)$ and $i_{j}$ respectively.
Note also that only the $p=M$ case of $\lambda^{[p]}$ appear and $\lambda^{[M]}=\lambda$.
Then \eqref{5-extended-wavefunction} is written as
\begin{multline} 
        \psi\Big(\bm z^{(1)}, \dots,\bm z^{(M)} \Big|\bm w^{(M)} \Big|
        k_1, \dots, k_M \Big| L_M \Big| \bm I \Big| \lambda\Big)
        \\
        = [\gamma]^{k_1 + \dots + k_M} \sum_{\sigma_1 \in S_{k_1}} \cdots \sum_{\sigma_M \in S_{k_M}}
        \prod_{p=1}^{M-1} \Bigg\{
            \prod_{a=1}^{k_p} \Bigg(
                \prod_{i=1}^{\widetilde{I}_{a}^{(p)}-1} 
                    \left[z^{(p)}_{\sigma_p(a)}- z^{(p+1)}_{\sigma_{p+1}(i)} \right]
\prod_{i=\widetilde{I}_{a}^{(p)}+1}^{k_{p+1}}
                    \left[z^{(p)}_{\sigma_p(a)}-z^{(p+1)}_{\sigma_{p+1}(i)}-\gamma \right]
                \\
                \times 
                \frac{\left[z^{(p)}_{\sigma_p(a)}-z^{(p+1)}_{\sigma_{p+1}(i)}
                +\lambda_{p+1}-\lambda_{i_{I_a^{(p)}}}+\gamma\left( C(I_a^{(p)},i_{I_a^{(p)}})-C(I_a^{(p)},p+1)-1\right)\right]}{
                    [\lambda_{i_{I_a^{(p)}}}-\lambda_{p+1}+\gamma(1-C(I_a^{(p)},i_{I_a^{(p)}}))]
                }
            \Bigg)
            \prod_{a<b}^{k_p} 
                \frac{[z^{(p)}_{\sigma_p(a)}-z^{(p)}_{\sigma_p(b)}+\gamma]}{[z^{(p)}_{\sigma_p(a)}-z^{(p)}_{\sigma_p(b)}]}
        \Bigg\}
        \\
        \times 
        \prod_{a=1}^{k_M} \Bigg(
            \prod_{i=1}^{I_a^{(M)}-1} 
                \left[z_{\sigma_M(a)}^{(M)}-w_i^{(M)}\right]
            \prod_{i=I_a^{(M)}+1}^{L_M} 
                \left[z_{\sigma_M(a)}^{(M)}-w_i^{(M)}-\gamma \right]
        \\
        \times 
        \frac{\left[z^{(M)}_{\sigma_M(a)}-w^{(M)}_{I_a}
        +\lambda_{M+1}-\lambda_{i_{I_a^{(M)}}}+\gamma\left( C(I_a^{(M)},i_{I_a^{(M)}})-C(I_a^{(M)},M+1)-1\right)\right]}{
            [\lambda_{i_{I_a^{(M)}}}-\lambda_{M+1}+\gamma(1-C(I_a^{(M)},i_{I_a^{(M)}}))]
        }
        \Bigg)
        \prod_{a<b}^{k_M} 
            \frac{[z^{(M)}_{\sigma_M(a)}-z^{(M)}_{\sigma_M(b)}+\gamma]}{[z^{(M)}_{\sigma_M(a)}-z^{(M)}_{\sigma_M(b)}]}.
\label{ellipticweightfunction}
    \end{multline}
We multiply \eqref{ellipticweightfunction} by some factors common to all summands of the sum and set ${\bm z}^{(M+1)}:={\bm w}^{(M)}$
to get the following normalized weight function
\begin{multline} 
        \psi^{\mathrm{norm}}\Big(\bm z^{(1)}, \dots,\bm z^{(M)} \Big|\bm z^{(M+1)} \Big|
        k_1, \dots, k_M \Big| L_M \Big| \bm I \Big| \lambda\Big)
        \\
        = \sum_{\sigma_1 \in S_{k_1}} \cdots \sum_{\sigma_M \in S_{k_M}}
        \prod_{p=1}^{M} \Bigg\{
            \prod_{a=1}^{k_p} \Bigg(
                \prod_{i=1}^{\widetilde{I}_{a}^{(p)}-1} 
                    \left[z^{(p)}_{\sigma_p(a)}- z^{(p+1)}_{\sigma_{p+1}(i)} \right]
\prod_{i=\widetilde{I}_{a}^{(p)}+1}^{k_{p+1}}
                    \left[z^{(p)}_{\sigma_p(a)}-z^{(p+1)}_{\sigma_{p+1}(i)}-\gamma \right]
                \\
                \times 
                \left[z^{(p)}_{\sigma_p(a)}-z^{(p+1)}_{\sigma_{p+1}(i)}
                +\lambda_{p+1}-\lambda_{i_{I_a^{(p)}}}+\gamma\left( C(I_a^{(p)},i_{I_a^{(p)}})-C(I_a^{(p)},p+1)-1\right)\right]
            \Bigg)
            \prod_{a<b}^{k_p} 
                \frac{[z^{(p)}_{\sigma_p(a)}-z^{(p)}_{\sigma_p(b)}+\gamma]}{[z^{(p)}_{\sigma_p(a)}-z^{(p)}_{\sigma_p(b)}]}
        \Bigg\}, \label{normalizedellipticweightfunction}
    \end{multline}
where
 \begin{align}
        C(k,l)=|\{j | i^{(M)}_j=l, j \leq k \}|, \label{dynamicalsymbolforweightfunction}
 \end{align}
 and $\sigma_{M+1}$ is defined to be the identity.

We next recall the version of elliptic weight functions by Konno \cite{Konno1}.
Note that the notations used below are basically different from the ones used in this paper.
The naive label of $I$ corresponds to tuple of integers
$(\mu_1,\dots,\mu_n), \mu_j \in \{ 1,\dots,N \}$.
We define the index set $I_l:=\{ i \in \{1,\dots,n \} \ | \ \mu_i=l \}$ ($l=1,\dots,N$)
and introduce tuples of nonnegative integers $\lambda:=(\lambda_1,\dots,\lambda_N)$ where $\lambda_l=|I_l|$.
The tuple of integers $(I_1,\dots,I_N)$ can be also regarded as another equivalent description for $I$.
Introduce $\lambda^{(l)}:=\lambda_1+\cdots+\lambda_l$
and $I^{(l)}:=I_1 \cup \cdots \cup I_l$. Note $\lambda^{(N)}=n$.
We also introduce notations for the elements
$I^{(l)}=\{ i_1^{(l)}<\cdots<i_{\lambda^{(l)}}^{(l)} \}$.
Replacing $\Pi$ by $\Pi^*$ in \cite[equations (5.1) and (5.2)]{Konno1} gives
\begin{align}
\widetilde{W}_I(t,z,\Pi^*)
&=\mathrm{Sym}_{t^{(1)}} \cdots \mathrm{Sym}_{t^{(N-1)}}
\widetilde{U}_I(t,z,\Pi^*), \\
\widetilde{U}_I(t,z,\Pi^*)
&=\prod_{l=1}^{N-1} \prod_{a=1}^{{\lambda}^{(l)}}
\Bigg(
\frac{[\nu_b^{(l+1)}-\nu_a^{(l)}+P_{\mu_s,l+1}-C_{\mu_s,l+1}(s)][1]}{[\nu_b^{(l+1)}-\nu_a^{(l)}+1][P_{\mu_s,l+1}-C_{\mu_s,l+1}(s)]} \Bigg|_{i_b^{(l+1)}=i_a^{(l)}=s}
\nonumber \\
&\times
\prod_{\substack{
b=1
\\
i_b^{(l+1)} > i_a^{(l)}
}}^{\lambda^{(l+1)}}
\frac{[\nu_b^{(l+1)}-\nu_a^{(l)}]}{[\nu_b^{(l+1)}-\nu_a^{(l)}+1]}
\prod_{b=a+1}^{\lambda^{(l)}} \frac{[\nu_a^{(l)}-\nu_b^{(l)}-1]}{[\nu_a^{(l)}-\nu_b^{(l)}]}
\Bigg), \label{ellipticweightfunctionKonno}
\end{align}
where we set $t_a^{(l)}=q^{2 \nu_a^{(l)}}$
$(l=1,\dots,N-1, a=1,\dots,\lambda^{(l)})$, $z_k=q^{2u_k}$
$(k=1,\dots,n)$, $\nu_s^{(N)}=u_s$ $(s=1,\dots,n)$ and
\begin{align}
C_{\mu_s,l+1}(s):=\sum_{j=s+1}^n \langle \overline{\epsilon}_{\mu_j},
h_{\mu_s,l+1} \rangle
=\sum_{j=s}^n \langle \overline{\epsilon}_{\mu_j},
h_{\mu_s} \rangle-
\sum_{j=s+1}^n \langle \overline{\epsilon}_{\mu_j},
h_{l+1} \rangle-1,
\end{align}
where $\langle \overline{\epsilon}_{j},h_k \rangle=\delta_{jk}$ for $\mathfrak{gl}_N$ case (rather than $\mathfrak{sl}_N$).
$P_{\mu_s,l+1}=P_{\mu_s}-P_{l+1}$ where $P_j$ are complex variables,
and $\mathrm{Sym}_{t^{(l)}}$ denotes symmetrization over the variables
$t_1^{(l)},\dots,t_{\lambda^{(l)}}^{(l)}$
\begin{align}
\mathrm{Sym}_{t^{(l)}} f(t_1^{(l)},\dots,t_{\lambda^{(l)}}^{(l)})
=\sum_{\sigma \in S_{\lambda^{(l)}}}
f(t_{\sigma(1)}^{(l)},\dots,t_{\sigma(\lambda^{(l)})}^{(l)}).
\end{align}

We multiply \eqref{ellipticweightfunctionKonno} by
$\prod_{l=1}^{N-1} \prod_{a=1}^{{\lambda}^{(l)}} \prod_{b=1}^{{\lambda}^{(l+1)}} [\nu_b^{(l+1)}-\nu_a^{(l)}+1]$
and some other factors to get
\begin{align}
\widetilde{W}_I^{\mathrm{norm}}(t,z,\Pi^*)
&=(-1)^{\sum_{l=1}^{N-1} |\lambda^{(l)}||\lambda^{(l+1)}|} \mathrm{Sym}_{t^{(1)}} \cdots \mathrm{Sym}_{t^{(N-1)}}
\widetilde{U}_I^{\mathrm{norm}}(t,z,\Pi^*), \\
\widetilde{U}_I^{\mathrm{norm}}(t,z,\Pi^*)
&=\prod_{l=1}^{N-1} \prod_{a=1}^{{\lambda}^{(l)}}
\Bigg(
[\nu_b^{(l+1)}-\nu_a^{(l)}+P_{\mu_s}-P_{l+1}-
\sum_{j=s}^n \langle \overline{\epsilon}_{\mu_j},
h_{\mu_s} \rangle+
\sum_{j=s+1}^n \langle \overline{\epsilon}_{\mu_j},
h_{l+1} \rangle+1
]|_{i_b^{(l+1)}=i_a^{(l)}=s}
\nonumber \\
&\times
\prod_{\substack{
b=1
\\
i_b^{(l+1)} > i_a^{(l)}
}}^{\lambda^{(l+1)}}
[\nu_b^{(l+1)}-\nu_a^{(l)}]
\prod_{\substack{
b=1
\\
i_b^{(l+1)} < i_a^{(l)}
}}^{\lambda^{(l+1)}}
[\nu_b^{(l+1)}-\nu_a^{(l)}+1]
\prod_{b=a+1}^{\lambda^{(l)}} \frac{[\nu_a^{(l)}-\nu_b^{(l)}-1]}{[\nu_a^{(l)}-\nu_b^{(l)}]}
\Bigg).
\end{align}

Since the quasi-period of the theta functions in real direction
is $1$ in Konno's notation where it is $\gamma$ in this paper,
we replace $\pm 1$ by $\pm \gamma$ to match with the version of the elliptic weight functions.
We also need to $i_a^{(l)} \to n+1-i_a^{(l)}$, $\mu_a \to \mu_{n+1-a}$, $s \to n+1-s$,
$\nu_a^{(l)} \to \nu_{\lambda^{(l)}+1-a}^{(l)}$,
which corresponds to reversing the labels of the quantum spaces.
Also note
$\sum_{j=1}^{s-1} \langle \overline{\epsilon}_{\mu_j},h_{l+1} \rangle=\sum_{j=1}^{s} \langle \overline{\epsilon}_{\mu_j},h_{l+1} \rangle$
when $i_b^{(l+1)}=i_a^{(l)}=s$ since $\mu_s \neq l+1$ in this case.
$\widetilde{W}_I^{\mathrm{norm}}(t,z,\Pi^*)$
can be rewritten as
\begin{align}
&\widetilde{W}_I^{\mathrm{norm}}(t,z,\Pi^*)=\mathrm{Sym}_{t^{(1)}} \cdots \mathrm{Sym}_{t^{(N-1)}}
\widetilde{U}_I^{\mathrm{norm}}(t,z,\Pi^*), \\
&\widetilde{U}_I^{\mathrm{norm}}(t,z,\Pi^*)
\nonumber \\
=& \prod_{l=1}^{N-1} \prod_{a=1}^{{\lambda}^{(l)}}
\Bigg(
[\nu_a^{(l)}-\nu_b^{(l+1)}+P_{l+1}-P_{\mu_s}
+\gamma \Bigg(
\sum_{j=1}^s \langle \overline{\epsilon}_{\mu_j},
h_{\mu_s} \rangle -1 \Bigg)-\gamma
\sum_{j=1}^{s} \langle \overline{\epsilon}_{\mu_j},h_{l+1} \rangle
]|_{i_b^{(l+1)}=i_a^{(l)}=s}
\nonumber \\
&\times
\prod_{\substack{
b=1
\\
i_b^{(l+1)} < i_a^{(l)}
}}^{\lambda^{(l+1)}}
[\nu_a^{(l)}-
\nu_b^{(l+1)}]
\prod_{\substack{
b=1
\\
i_b^{(l+1)} > i_a^{(l)}
}}^{\lambda^{(l+1)}}
[\nu_a^{(l)}-\nu_b^{(l+1)}-\gamma]
\prod_{b=a+1}^{\lambda^{(l)}} \frac{[\nu_a^{(l)}-\nu_b^{(l)}+\gamma]}{[\nu_a^{(l)}-\nu_b^{(l)}]}
\Bigg).
\end{align}

Now, note there is the following correspondence between the symbols
   \begin{gather*}
   N \leftrightarrow M+1
\\
        \stackrel{Konno}{I^{(j)}}  \leftrightarrow {\bm I}^{(j)}_{k_j}
        \\
\stackrel{Konno}{i_a^{(j)}} \leftrightarrow I^{(j)}_a
\\
\stackrel{Konno}{\lambda^{(j)}} \leftrightarrow k_j, j=1,\dots,M
\\
\stackrel{Konno}{\lambda^{(N)}} \leftrightarrow L_M
       \\
\stackrel{Konno}{P_j} \leftrightarrow \lambda_j
\\
\stackrel{Konno}{\nu_a^{(j)}} \leftrightarrow z_a^{(j)}
\\
\sum_{j=1}^s \langle \overline{\epsilon}_{\mu_j},
h_{\mu_s} \rangle -1 
\Bigg|_{i_b^{(l+1)}=i_a^{(l)}=s}
\leftrightarrow C(I_a^{(l)},i_{I_a^{(l)}})-1
\\
\sum_{j=1}^{s} \langle \overline{\epsilon}_{\mu_j},h_{l+1} \rangle
\Bigg|_{i_b^{(l+1)}=i_a^{(l)}=s}
\leftrightarrow C(I_a^{(l)},l+1)
    \end{gather*}
from which we note the elliptic weight functions are equivalent.

Note that to see the correspondence, we use the following equivalence
\begin{align}
i_b^{(l+1)}=i_a^{(l)} \longleftrightarrow b=\widetilde{I}_a^{(l)}, 
\ \ \
i_b^{(l+1)}>i_a^{(l)} \longleftrightarrow b>\widetilde{I}_a^{(l)},
\ \ \
i_b^{(l+1)}<i_a^{(l)} \longleftrightarrow b<\widetilde{I}_a^{(l)}, 
\end{align}
between the notations used in this paper and those
in Konno and also Rimanyi-Tarasov-Varchenko.
For example, $i_b^{(2)}=i_a^{(1)}$
means that $I_b^{(2)}$, which
labels the position of the $b$-th place which is colored either by color 1 or
color 2, is actually colored by 1, and it is the $a$-th place which is colored by 1.
Noting that $\widetilde{{\bm I}}_{k_1}^{(1)}$ is the set induced, we find that $i_b^{(2)}=i_a^{(1)}$
corresponds to $b=\widetilde{I}_a^{(1)}$.

Finally, we give below the presentation of $\psi^{\mathrm{norm}}\Big(\bm z^{(1)}, \dots,\bm z^{(M)} \Big|\bm z^{(M+1)} \Big|
        k_1, \dots, k_M \Big| L_M \Big| \bm I \Big| \lambda\Big)$
and $\widetilde{W}_I^{\mathrm{norm}}(t,z,\Pi^*)$
using the conventions of Rimanyi-Tarasov-Varchenko
\cite[equations (2.9), (2.10)]{RTV2}
\begin{align}
W_{I}^{\mathrm{ell},\mathrm{norm}}({\bm t},{\bm z},h,{\bm \mu})
=
\mathrm{Sym}_{t^{(1)}} \dots \mathrm{Sym}_{t^{(N-1)}}
U_I^{\mathrm{ell},\mathrm{norm}}({\bm t}, {\bm z},h, {\bm \mu}),
\label{RTVellipticweightfunctions}
\end{align}
\begin{align}
U_I^{\mathrm{ell},\mathrm{norm}}({\bm t}, {\bm z},h, {\bm \mu})=
\prod_{k=1}^{N-1} \prod_{a=1}^{\lambda^{(k)}}
\Bigg(
\prod_{c=1}^{\lambda^{(k+1)}}
\psi_{I,k,a,c}^{\mathrm{ell}}(t_a^{(k)}/t_c^{(k+1)})
\prod_{b=a+1}^{\lambda^{(k)}} \frac{\vartheta(ht_a^{(k)}/t_b^{(k)})}
{\vartheta(t_b^{(k)}/t_a^{(k)})}
\Bigg), \label{component}
\end{align}
where
\begin{align}
\psi_{I,k,a,c}^{\mathrm{ell}}(x)
=
\begin{cases}
    \vartheta(x), & i_c^{(k+1)}<i_a^{(k)}, \\
   \displaystyle \vartheta(
x h^{p_{I,j(I,k,a)}(i_a^{(k)})-p_{I,k+1}(i_a^{(k)})} \mu_{k+1}/\mu_{j(I,k,a)}
),  & i_c^{(k+1)}=i_a^{(k)}, \\
     \vartheta(xh^{-1}),  & i_c^{(k+1)}>i_a^{(k)},
\end{cases}.
\label{cases}
\end{align}
Here we use the multiplicative version of the theta function
\begin{align}
\vartheta(x)=(x^{1/2}-x^{-1/2})\phi(qx)\phi(q/x),
\ \ \ \phi(x)=\prod_{s=0}^\infty (1-q^s x),
\end{align}
and $j(I,k,a) \in \{1,\dots,N \}$ is such that
$i_a^{(k)} \in I_{j(I,k,a)}$ and
\begin{align}
p_{I,j}(m)=|I_j \cap \{1,\dots,m-1 \}|, \ \ \ j=1,\dots,N.
\end{align}
Note $\psi_{I,k,a,c}^{\mathrm{ell}}(x)$ is obtained from the same symbol
by Rimanyi-Tarasov-Varchenko by replacing $x$ by $xh^{-1}$, and the $t$-variables are replaced by their inverses.
The correspondence between the symbols used in this paper and \cite{RTV2} is the following.
  \begin{gather*}
   N \leftrightarrow M+1
\\
        \stackrel{RTV}{I^{(j)}}  \leftrightarrow {\bm I}^{(j)}_{k_j}
        \\
\stackrel{RTV}{i_a^{(j)}} \leftrightarrow I^{(j)}_a
\\
\stackrel{RTV}{\lambda^{(j)}} \leftrightarrow k_j, j=1,\dots,M
\\
\stackrel{RTV}{\lambda^{(N)}} \leftrightarrow L_M
       \\
\stackrel{RTV}{h} (\mathrm{multiplicative}) \leftrightarrow \gamma \ (\mathrm{additive})
\\
\stackrel{RTV}{\mu_j}   (\mathrm{multiplicative})   \leftrightarrow \lambda_j \ (\mathrm{additive})
\\
\stackrel{RTV}{t_a^{(j)}}  (\mathrm{multiplicative})   \leftrightarrow z_a^{(j)} \ (\mathrm{additive})
        \\
        p_{I,j(I,k,a)} (\stackrel{RTV}{i_a^{(k)}} ) \leftrightarrow C^{(M)}(I^{(k)}_a,i_{I^{(k)}_a}) - 1
        \\
        p_{I,k+1} (\stackrel{RTV}{i_a^{(k)}}) \leftrightarrow C^{(M)}(I^{(k)}_a,k+1)
    \end{gather*}







\section*{Conclusion}

We have introduced a family of $\mathfrak{gl}_{M+1}$-type multisymmetric special functions in the rational, trigonometric and elliptic cases and shown, using nested Izergin-Korepin analysis, that these special functions are exactly equal to a generalization of the lattice partition function, or off-shell nested Bethe wavefunction. 
This work generalizes the result of Foda-Manabe \cite{FM}, as well as the elliptic multisymmetric functions in Konno \cite{Konno1,Konno2} and Rim{\'a}nyi-Tarasov-Varchenko \cite{RTV2}. 
It would be of interest to analyse how this generalisation of the Bethe wavefunction might make sense in the language of quiver varieties, 
or the interpretation in the Bethe-Gauge correspondence.
Some cases beyond off-shell Bethe wavefunctions
were discussed in \cite{FM} by a limiting procedure.
However, note that for the elliptic case involving theta functions,
the limiting procedure as given in \cite{FM} cannot be applied.

We imagine that these results extend to the supersymmetric versions of these functions in the standard way.
A further direction of research is to investigate the nested Bethe wavefunction associated with other simple Lie groups in the rational and trigonometric types beyond rank 1, of which significantly less is known \cite{LP} \cite{GR}.

\section*{Acknowledgements}
This work was partially supported by Grants-in-Aid for Scientific Research (C) 20K03793, (C) 21K03176 and (C) 24K06889.

\end{document}